\documentclass[a4paper, 11pt]{amsart}
               		
\usepackage{morefloats}
\usepackage{graphicx,multicol}%
\usepackage{lipsum}
\usepackage{amsmath}
\usepackage{graphicx, caption}
\usepackage{amsthm}
\usepackage{subcaption}
\usepackage[colorinlistoftodos]{todonotes}
\usepackage{multirow}
\usepackage{amssymb}
\usepackage{bbm}
\usepackage{verbatim}
\usepackage{diagbox}
\usepackage{tikz}
\usepackage{float}
\usepackage{pst-eucl}
\usepackage{pst-xkey}
\usepackage{pstricks-add} 
\usepackage{tcolorbox}
\usetikzlibrary{shapes,snakes}
\usepackage{pst-plot}
\usepackage{xcolor}
\usepackage{pst-text}
\usepackage{pst-all}
\usepackage[titletoc,title]{appendix}
\newtheorem{theorem}{Theorem}
\newtheorem{lemma}[theorem]{Lemma}
\newtheorem{mydef}{Definition}

\newcommand{\argmin}{argmin}
\newcommand{\Var}{Var}

\DeclareRobustCommand\sbseries{\not@math@alphabet\sbseries\mathbf\fontseries\sbdefault\selectfont}
\DeclareTextFontCommand{\textsb}{\sbseries}
\DeclareRobustCommand\sfitseries{\not@math@alphabet\sfitseries\normalfont\fontseries{m}\fontshape{it}\selectfont}
\DeclareTextFontCommand{\textsfi}{\sfitseries}
\DeclareOldFontCommand{\rm}{\normalfont\rmfamily}{\mathrm}
\DeclareOldFontCommand{\sf}{\normalfont\sffamily}{\mathsf}
\DeclareOldFontCommand{\tt}{\normalfont\ttfamily}{\mathtt}
\DeclareOldFontCommand{\bf}{\normalfont\bfseries}{\mathbf}
\DeclareOldFontCommand{\it}{\normalfont\itshape}{\mathit}
\DeclareOldFontCommand{\sl}{\normalfont\slshape}{\@nomath\sl}
\DeclareOldFontCommand{\sc}{\normalfont\scshape}{\@nomath\sc}
\DeclareRobustCommand*\cal{\mathcal}

\begin{document}

\title{Adaptive multicenter designs for continuous response clinical trials in the presence of an unknown sensitive subgroup}
\title[Adaptive designs in the presence of unknown sensitive subgroup]{Adaptive multicenter designs for continuous response clinical trials in the presence of an unknown sensitive subgroup}

\author{Daria Rukina}
\newcommand{\Addresses}{{
  \bigskip
  \footnotesize
  D.~Rukina, \textsc{Institute of Mathematics, \'Ecole Polytechnique F\'ed\'erale de Lausanne (EPFL),
    Lausanne, Switzerland 1015}\par\nopagebreak
  \textit{E-mail address}, D.~Rukina: \texttt{daria.rukina@epfl.ch} }}

\begin{abstract}The partial effectiveness of drugs is of importance to the pharmaceutical industry. Randomized controlled trials (RCTs) assuming the existence of a subgroup sensitive to the treatment are already used. These designs, however, are available only if there is a known marker for identifying subjects in the subgroup. In this paper we investigate a model in which the response in the treatment group $Z^T$ has a two-component mixture density $(1-p)\mathcal N(\mu^C, \sigma^2)+p\mathcal N(\mu^T, \sigma^2)$ representing the treatment responses of \emph{placebo responders} and \emph{drug responders}. The treatment-specific effect is $\mu = \frac{\mu^T-\mu^C}{\sigma}$ and $p$ is the prevalence of the drug responders in the population. Other patients in the treatment group react as if they had received a placebo.
\par We develop one- and two-stage RCT designs that are able to detect a sensitive subgroup based solely on the responses. We also extend them to a multicenter RCTs using Hochberg's step-up procedure. We avoid extensive simulations and use simple and quick numerical optimization methods.
\end{abstract}

\maketitle

\section{Introduction}

Because of population heterogeneity, treatment effects may vary considerably from one subject to another. In \cite{incentive2004} it is reported that on average only 60\% of treated patients react to prescription drugs, whereas the others either do not benefit or have adverse effects. If a trial showed no significant \emph{average treatment effect}, it can be wrong to conclude that the treatment effect does not exist and that the treatment is futile. There may well be a subpopulation of subjects for whom the treatment is effective. As an example, consider the Clomethiazole Acute Stroke Study \cite{Wahlgren21} where the relative functional independence of patients with acute hemispheric stroke was investigated. The trial results showed no significant difference between the clomethiazole and the placebo groups (56.1\% vs. 54.8\%) with a p-value of $0.649$. However, in the group of patients with TACS (total anterior circulation syndrome), 40.8\% reached relative functional independence in contrast to 29.8\% of TACS patients in the placebo group, with a p-value of $0.008$.

The idea that there is a proportion of non-responders in the treatment group is often highlighted in clinical papers \cite{Rubin:1974aa, Frangakis:2002aa, Muthen:2009aa, Leiby:2009aa, He:2014aa}. If one has specific biomarkers which point to the subjects who are likely to respond to an intervention, then a subgroup analysis is usually conducted in addition to the standard procedure. Most of the time, however, there is no such information about the sensitive subgroup, and a standard designs may fail to prove the drug's efficacy. 
\par Moreover, even if the proportion of the sensitive patients is big, there is another hurdle which can mask a drug-specific effect, namely the placebo. When its effect is comparable with a drug, the detection of a subgroup is complicated substantially. This problem may occur in treatments for neurological disorders, where many trials fail to show efficacy, and the reason is claimed to be a large placebo effect \cite{Agid:2013aa, Tuttle:2015aa, Holmes:2016aa}. For our purpose, the placebo effect is a combination of natural course of the disease, the Hawthorne effect and other non-specific effects influencing the treatment outcome.
\par Given the assumptions above, we elaborate a framework in which we are able to detect the sensitive subgroup based on the following model. Let the response to the drug, $Z$, be a normally distributed random variable. We assume that there are two groups of patients: placebo-only responders with $Z\sim \mathcal N\left(\mu^C, \sigma^2\right)$ and placebo-and-drug responders with $Z\sim \mathcal N\left(\mu^T, \sigma^2\right)$, where $\mu^T > \mu^C$. The drug-specific effect, $\mu^T - \mu^C$, and the placebo effect are additive. The response in the treatment group is then modelled as a mixture (see Fig.~\ref{diff_scheme}) $$Z^T \sim (1-p)\mathcal N\left(\mu^C, \sigma^2\right)+ p \mathcal N\left(\mu^T, \sigma^2\right), \hspace{0.5cm}$$ 
where $p \ge 0$ is the proportion of placebo-and-drug responders (the sensitive subpopulation). The design is constructed to determine whether $p>0$. Denote the standardized drug-specific effect as $\mu = \frac{\mu^T-\mu^C}{\sigma}$. Given the standardized estimates of responses in the treatment group, $$X_i = \frac{Z^T_i-\overline{Z^C}}{\sigma}, \hspace{0.1cm} i = 1, \ldots, n,$$ the average value $\overline X$ has expectation $\mu p$ and is distributed as $\sum_{k=0}^n \binom {n} {k}p^k(1-p)^{n-k} \mathcal{N}\left(\frac{k}{n}\mu,\frac{2}{n}\right)$. The test based on the average of observations can be used in deciding between $\mathcal H_0: p=0$ against $\mathcal H_1: p>0$. Along with its simplicity the mean value statistic can also be used with non-normal responses that may be encountered in clinical trials.

To control the power of detection, we consider the minimal power over a region of strong effect, the set of $(\mu, p)$ that are deemed to be of interest. We propose one- and two-stage designs for the single and multicenter trials. This framework is mainly aimed at the early phase II before the identification of the drug-responders subgroup.

\par The outline of the paper is as follows. In Section 2, we introduce the model, define type I and type II errors, and describe the designs. In Section 3, we discuss the choice of parameters and optimal designs. We present the planning procedure, from the user (study designer) point of view, with the plots and tables to guide the choice of parameters. In Section 4, the performance of the study designs is illustrated. Section 5 provides an additional argumentation and discusses further related questions, such as estimation of $\mu$ and $p$. Section 6 concludes.
\begin{figure}
\begin{minipage}{.5\textwidth}%
\centering
\includegraphics[width = 7cm]{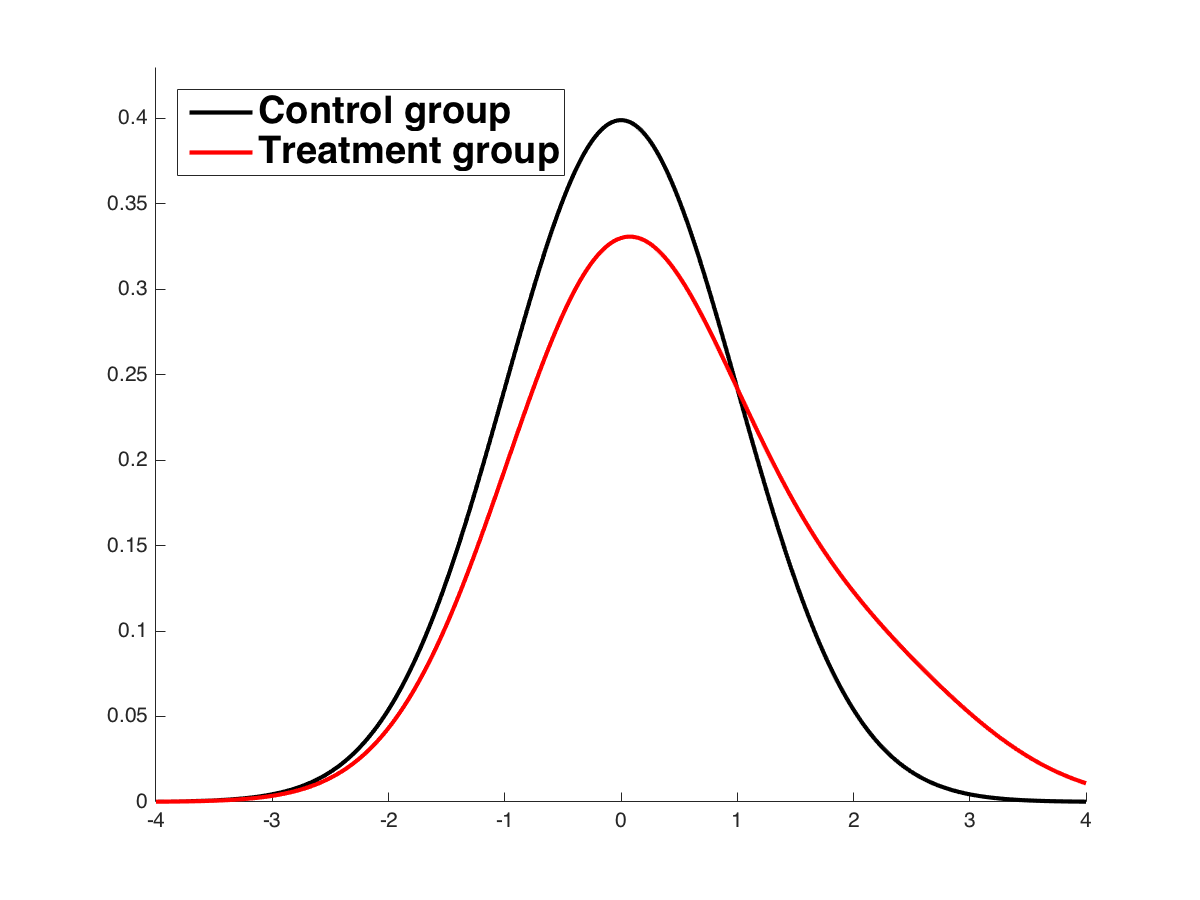}
\end{minipage}%
\begin{minipage}{.5\textwidth}%
\centering
\includegraphics[width = 7cm]{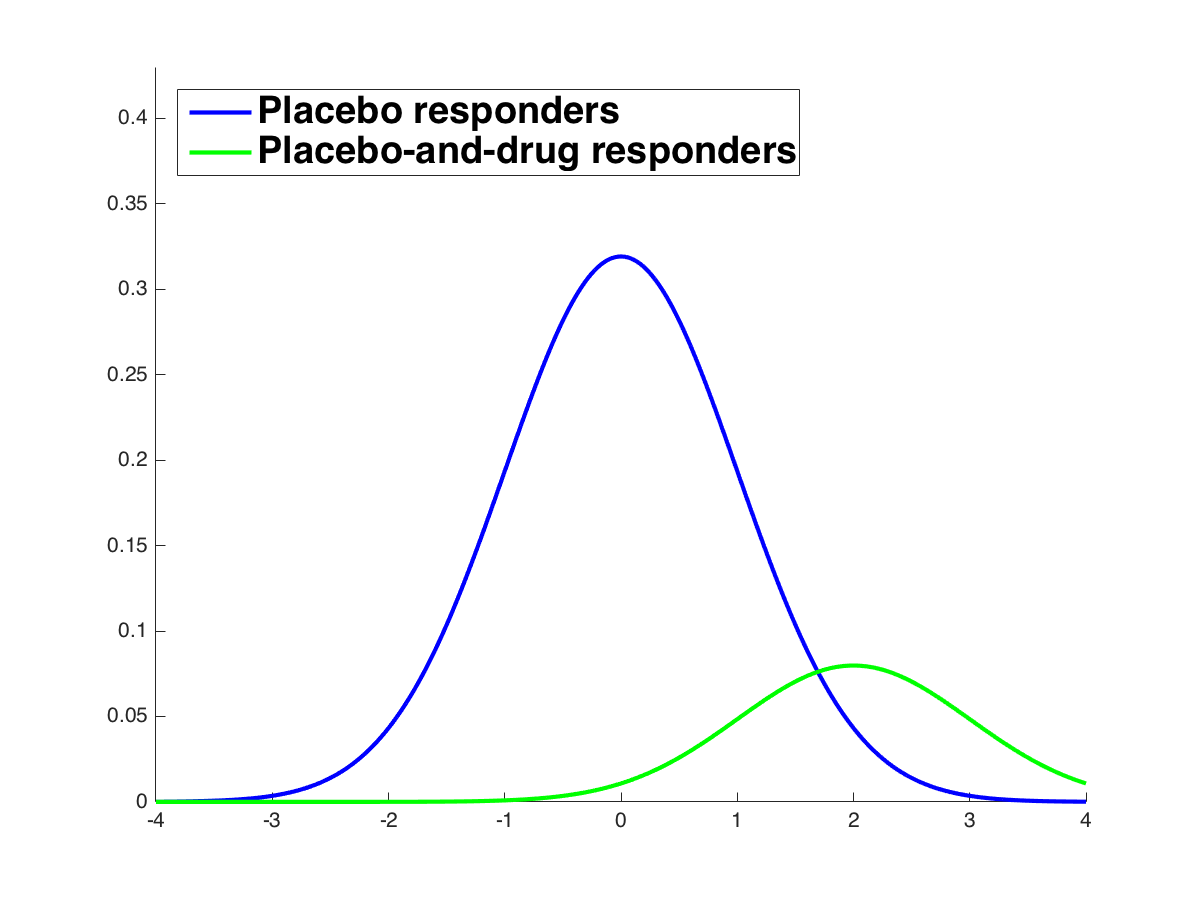}
\end{minipage} %
\caption{\hspace{0.2cm}Modelling drug response as a mixture of two continuous distributions has been proposed to be more reasonable in RCT for some types of indications. On the left, the control group is centered at zero and the mixture created in the treatment group is asymmetric with an enlarged upper tail. If we could separate the responders and non-responders in the treatment group, we would observe the graphs as in the right hand side.}
\label{diff_scheme}
\end{figure}
\section{The designs}
In all RCTs discussed in this chapter we consider any treatment group at any stage to have a corresponding control group of the same size. Let $Z^T_1,  \ldots , Z^T_n$ and $Z^C_1,  \ldots , Z^C_n$ be the responses from the treatment and the control groups, respectively. For convenience we define normalized responses as
\begin{align*}
Y^C = \frac{Z^C - \mu^C}{\sigma} \sim\mathcal N(0,1), \hspace{0.2cm} Y^T = \frac{Z^T - \mu^C}{\sigma} \sim (1-p)\mathcal N(0,1)+p\mathcal N(\mu,1),
\end{align*}
where $\mu=\frac{\mu^T - \mu^C}{\sigma}$. Since we do not know $\mu^C$ and $\sigma$, it is not possible to observe $Y^T$. Instead, we take its estimate, $X_i = \frac{Z_i^T-\hat\mu^C}{\sigma}$, where $\hat\mu^C = \frac{\sum_{i=1}^n Z_i^C}{n}$. For the design development we assume $\sigma$ to be known, but in practice we may use the approximation $\hat\sigma^C = \sqrt{\frac{\sum_{i=1}^n (Z^C_i-\hat \mu^C)^2}{n-1}}$. In simulations shown later, we show that the approximation error has a negligible effect on the results.

\par In a simple RCT with a single center, the test is based on the mean value statistic, $$\overline X = \frac{X_1+\ldots+X_n}{n} = \overline{Y^T} - \frac{\hat \mu^C-\mu^C}{\sigma}.$$ Note that $$\overline{Y^T} \sim \sum_{k=0}^n \binom {n} {k}p^k(1-p)^{n-k} \mathcal{N}\left(\frac{k}{n}\mu,\frac{1}{n}\right),$$ $\frac{\hat \mu^C-\mu^C}{\sigma} \sim \mathcal N\left(0, \frac{1}{n}\right)$, and the distribution of $\overline X$ is
\begin{align}\label{main_expr}
&\overline X \sim \sum_{k=0}^n \binom {n} {k}p^k(1-p)^{n-k} \mathcal{N}\left(\frac{k}{n}\mu,\frac{2}{n}\right)\nonumber.
\end{align}

\subsection{Hypothesis testing}
The goal of a trial is to determine if there is a sensitive subpopulation. The decision is based on a statistical test of
\begin{equation}
\label{main_hypothesis}
\begin{aligned}
&\mathcal H_0: Y^T \sim \mathcal N(0,1)\hspace{1.8cm} \text{versus}\\
&\mathcal H_1: Y^T \sim (1-p)\mathcal{N}(0,1)+p\mathcal{N}(\mu,1), \hspace{1cm} p\in(0,1], \mu>0.
\end{aligned}
\end{equation}
Of course, not all pairs $(\mu,p)$ are of equal interest for the drug developer. It is of small interest if $\mathcal H_0$ is rejected when $\mathcal H_1$ is true with some small $p$ and/or $\mu$. Therefore, it is important to define \emph{interesting} values of $(\mu, p)$, for which one wants to reject the null. 
\begin{mydef}\label{def:se}
The region $\cal E$ of \emph{strong effect} (se) is the set of \emph{interesting} pairs $(\mu, p)$ such that $\mu \ge \mu_i$ if $p\in[p_i, p_{i+1}]$, where  $\mu_s <\mu_{s-1}<\ldots<\mu_1$ and $0<p_1<\ldots<p_{s+1} = 1$ (see Fig.~\ref{fig:region_of_interest}).
\end{mydef}


\begin{figure}
\begin{tikzpicture}
\draw[thick,->] (0,0)--(5,0) node[right]{$\mu$};
\draw[thick,->] (0,0)--(0,5) node[above]{$p$};
\draw[draw = white,fill=gray!20!white]  (0.5,4)--(0.5,3)--(1.3,3)--(1.3,2.5)--(2.8,2.5)--(2.8,1)--(5,1)--(5,4)--cycle;
\draw[line width=0.3mm] (5,4)--(0.5,4);
\draw[line width=0.3mm] (0.5,4)--(0.5,3);
\draw[line width=0.3mm] (0.5,3)--(1.3,3);
\draw[line width=0.3mm] (1.3,3)--(1.3,2.5);
\draw[line width=0.3mm] (1.3,2.5)--(2.8,2.5);
\draw[line width=0.3mm] (2.8,2.5)--(2.8,1);
\draw[line width=0.3mm] (2.8,1)--(5,1);
\draw[dotted, line width=0.3mm] (0,4)--(0.5,4);
\draw[dotted, line width=0.3mm] (0,3)--(0.5,3);
\draw[dotted, line width=0.3mm] (0.5,3)--(0.5,0);
\draw[dotted, line width=0.3mm] (0,2.5)--(1.3,2.5);
\draw[dotted, line width=0.3mm] (1.3,2.5)--(1.3,0);
\draw[dotted, line width=0.3mm] (0,1)--(2.8,1);
\draw[dotted, line width=0.3mm] (2.8,1)--(2.8,0);
\draw(-0.2,4) node{1};
\draw(0, 4) node[inner sep=1pt,fill,circle]{};
\draw(-0.2,3) node{$p_s$};
\draw(0, 3) node[inner sep=1pt,fill,circle]{};
\draw(0.5,-0.2) node{$\mu_s$};
\draw(0.5, 0) node[inner sep=1pt,fill,circle]{};
\draw(-0.2,2.5) node{$p_i$};
\draw(0, 2.5) node[inner sep=1pt,fill,circle]{};
\draw(1.3,-0.2) node{$\mu_i$};
\draw(1.3, 0) node[inner sep=1pt,fill,circle]{};
\draw(-0.2,1) node{$p_1$};
\draw(0, 1) node[inner sep=1pt,fill,circle]{};
\draw(2.8,-0.2) node{$\mu_1$};
\draw(2.8, 0) node[inner sep=1pt,fill,circle]{};
\draw(3,3) node{strong effect};
\draw(3.5, 1.5) node{${\cal E}$};
\end{tikzpicture}
\caption{\hspace{0.2cm}The region of strong effect, $\mathcal E$, is a subset of the $(\mu,p)$ plane satisfying $\mu\ge\mu_i$ if $p\in[p_i,p_{i+1}]$ where $\mu_s <\mu_{s-1}<...<\mu_1$ and $0<p_1<p_2<\ldots p_{s+1}=1$.}
\label{fig:region_of_interest}
\end{figure}
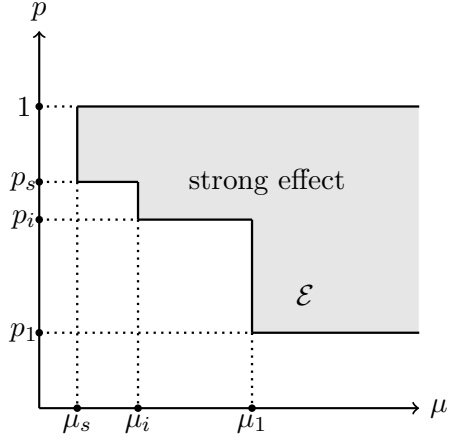

In the next two subsections, we introduce the single and the multicenter RCT designs, for which type I and II errors are defined as follows.
\begin{mydef}\label{def:errors}
For a simple study, a type I error occurs when a true null is rejected. A type II error occurs if a false null is not rejected in the presence of a strong effect. The corresponding type I error rate will be denoted as $\alpha$ and the maximum of the type II error rate over $\cal E$ is denoted as $\beta^{se}$. In the case of a multicenter study with $M$ centers a type I error occurs if at least one of the true nulls is rejected, and a type II error occurs if at least $m$ false nulls among the $M_1$ centers with strong effects are not rejected. Note that $m\le M_1\le M$. The type I error is a family-wise error, while the type II error depends on $M_1$ and $m$. The type II error rate maximized over $\cal E$ will be denoted by $\beta^{se}$ for the single stage study and $\beta^{se}_{\text{fw}}(M_1, m)$ of the multicenter study.
\end{mydef}
\subsection{Simple RCT designs - single center RCT}
We start with the simplest design: single center, one stage. For the mean value statistic the hypothesis testing problem is
\begin{equation}\label{distribution}
\begin{aligned}
&\mathcal H_0: \overline X \sim \mathcal{N}\left(0,\frac{2}{n}\right) ;\\
&\mathcal H_1: \overline X \sim \sum_{k=0}^n \binom {n} {k}p^k(1-p)^{n-k} \mathcal{N}\left(\frac{k}{n}\mu,\frac{2}{n}\right).
\end {aligned}
\end{equation}
The design of a trial is defined by the sample size $n$ and some positive threshold $\eta$, upon which the decision about a subgroup's existence is made, 
\begin{equation}\tag{D1}
\begin{aligned}
\text{Reject } \mathcal H_0 \text{ if } \overline X >\eta.
\end{aligned}
\end{equation}
The chosen test is justified by the following lemma:
\begin{lemma}\label{lemma:ump}
The tests that reject $\mathcal H_0$ for $\overline X>\eta$ are uniformly most powerful.
\end{lemma}
The equivalent formulation of the design is
\begin{equation}\tag{D1}
\begin{aligned}
\text{Reject } \mathcal H_0 \text{ if the } \text{p-value}\left(\overline X\right) < \alpha, \text{ where } \text{p-value}\left(\overline X\right) = 1-\Phi\left(\overline X \sqrt{\frac{n}{2}} \right) \text{ and } \alpha<0.5.
\end{aligned}
\end{equation}
Here, $\alpha$ is a type I error rate, and $\eta=z_{1-\alpha}\sqrt{\frac{2}{n}}$, where $z_{1-\alpha }= \Phi^{-1}(1-\alpha)$ stands for the standard Gaussian quantile. The probability of a false negative for a fixed alternative is  
\begin{equation}\label{eq:beta}
\begin{aligned}
&\beta(n,\eta,\mu,p)=\sum_{k=0}^n \binom {n} {k}p^k(1-p)^{n-k}\Phi\left( \left(\eta-\frac{k}{n}\mu\right)\sqrt{\frac{n}{2}}\right).
\end{aligned}
\end{equation}
By Definition \ref{def:errors},
 \begin{equation}\label{eq:beta^se}
\beta^{se}(n,\alpha) = \underset{(\mu,p)\in \cal E}{\max}\beta\left(n,z_{1-\alpha}\sqrt{\frac{2}{n}},\mu,p\right).
\end{equation}

As the following lemma shows, one can easily compute $\beta^{se}(n,\alpha)$ for the regions $\mathcal E$ of the form we defined.
\begin{lemma}\label{lemma:beta_max}
For the region of strong effect defined as in Definition \ref{def:se}, the maximum type II error rate for the one-stage RCT is
$$ \beta^{se}(n,\alpha) =\underset{i=1, ..., s}{\max}\beta\left(n,z_{1-\alpha}\sqrt{\frac{2}{n}},\mu_i,p_i\right).$$
\end{lemma}
\subsubsection{Two stage designs}
 By Wald's theory of sequential analysis, there are modifications of the RCT design enabling significant reduction in the required minimum number of the participants. The idea of sequential analysis is to split the testing process into several steps and update the inference based on the information accumulated from the previous steps \cite{ARMITAGE:1960aa}. In the context of RCT, the steps are the stages of the trial. Having many stages requires additional parameters, which might be difficult to choose. Therefore, in this paper, we consider only the extension to two-stage designs. We will show how the introduction of the second stage decreases the total expected number of enrolled patients in comparison with the one-stage design, while maintaining the same error rates.
 
For the first stage we propose two stopping rules. First, if the evidence in favour of the alternative is very strong after the first stage, there is no need to conduct the second stage. Second, if the evidence in favour of the null is strong enough at the first stage, one can stop the trial early, claiming the absence of an effect.

To start with, consider some fixed sample sizes with $n_1$ and $n_2$ participants at the first and second stages, respectively. We also assume that there are two different control groups at each stage with the same number of participants as in the treatment groups. Denote the mean observed values after the first and the second stages by $\overline X_1$ and $\overline X_2$, and the mean of all observations as $\overline X$. The design is described as follows:
\begin{equation}\tag{D2}
\begin{aligned}
& \bullet \text{ If } \overline X_1 < \eta_0:\, \mathcal H_0 \text{ is not rejected and the trial is stopped }\\
& \hspace{2.3cm} \text{(strong evidence of an absence of a positive treatment effect)}; \\
& \bullet \text{ If }\overline X_1 > \eta_1:\, \text{ reject }\mathcal H_0 \text{ and stop the trial (strong evidence in favour of } \mathcal H_1);\\
& \bullet \text{ If }\eta_0 \le \overline X_1 \le \eta_1: \text{ conduct the second stage and reject } \mathcal H_0 \text{ if }\overline X>\eta_2  \\
& \hspace{3.1cm}(\text{strong evidence in favour of }\mathcal H_1 \text{ after two steps}).
\end{aligned}
\end{equation}
\vspace{0.3cm}
The rejection region is depicted in Fig. \ref{rejection_region} (grey area).
For the thresholds, $\eta_0, \, \eta_1$, denote the corresponding probabilities $\mathbb P\left(\overline{X_1} <\eta_0\right)$ and  $\mathbb P\left(\overline{X_1} >\eta_1\right)$ under the null as $\alpha_0 =  \Phi\left(\eta_0\sqrt{\frac{n_1}{2}}\right)$ and $\alpha_1 = 1- \Phi\left(\eta_1\sqrt{\frac{n_1}{2}}\right)$.

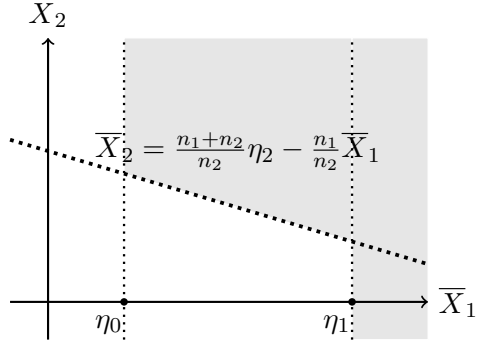
\begin{figure}
\centering
\begin{tikzpicture}
\draw[draw = white,fill=gray!20!white]  (1,3.5)--(1,1.7)--(4,0.8)--(4,3.5)--cycle;
\draw[draw = white,fill=gray!20!white]  (4,3.5)--(5,3.5)--(5,-0.5)--(4,-0.5)--cycle;
\draw[thick,->] (-0.5,0)--(5,0) node[right]{$\overline X_1$};
\draw[thick,->] (0,-0.5)--(0,3.5) node[above]{$\overline X_2$};   
\draw[domain=-0.5:5,dotted, line width=0.5mm] plot (\x,{2-0.3*\x});
\draw[dotted, line width=0.3mm] (1,-0.5)--(1,3.5);
\draw[dotted, line width=0.3mm] (4,-0.5)--(4,3.5);
\draw(2.5,2) node{$\overline X_2 = \frac{n_1+n_2}{n_2}\eta_2-\frac{n_1}{n_2}\overline X_1$};
\draw(0.8,-0.3) node{$\eta_0$};
\draw(1, 0) node[inner sep=1pt,fill,circle]{};
\draw(3.8,-0.3) node{$\eta_1$};
\draw(4, 0) node[inner sep=1pt,fill,circle]{};
\end{tikzpicture}
\caption{\hspace{0.2cm}Scheme of the trial performed in two stages. Rejection region is depicted in grey.}
\label{rejection_region}
\end{figure} 

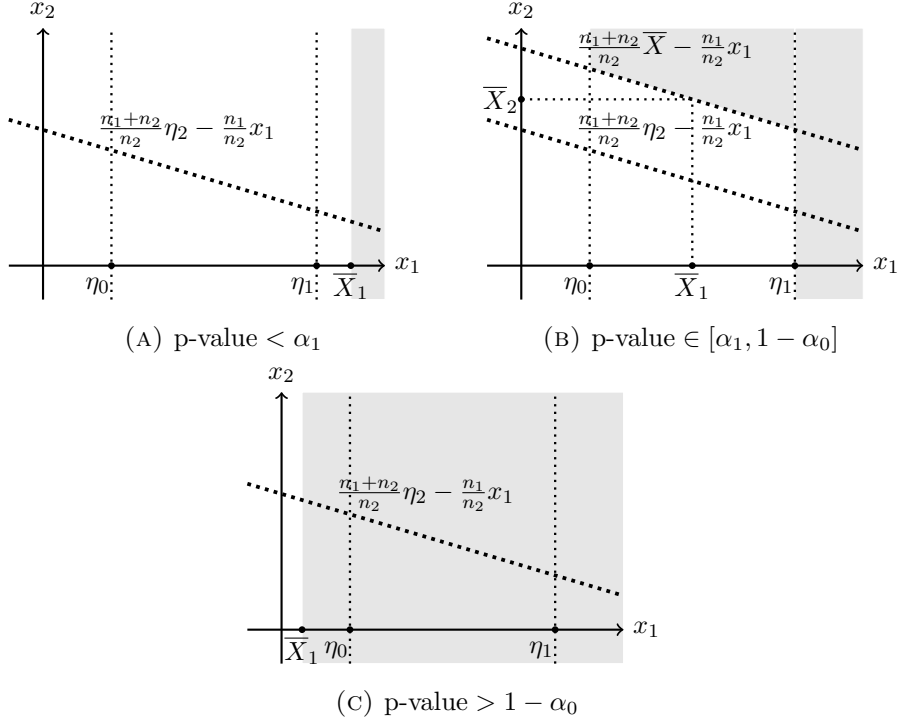
\begin{figure}
\centering
\begin{subfigure}{.45\textwidth}
\begin{tikzpicture}[scale=0.9, every node/.style={scale=0.9}]
\draw[draw = white,fill=gray!20!white]  (4.5,-0.5)--(4.5,3.5)--(5,3.5)--(5,-0.5)--cycle;
\draw[thick,->] (-0.5,0)--(5,0) node[right]{$x_1$};
\draw[thick,->] (0,-0.5)--(0,3.5) node[above]{$x_2$};   
\draw[domain=-0.5:5,dotted, line width=0.5mm] plot (\x,{2-0.3*\x});
\draw[dotted, line width=0.3mm] (1,-0.5)--(1,3.5);
\draw[dotted, line width=0.3mm] (4,-0.5)--(4,3.5);
\draw(2.1,2) node{$\frac{n_1+n_2}{n_2}\eta_2-\frac{n_1}{n_2}x_1$};
\draw(0.8,-0.3) node{$\eta_0$};
\draw(1, 0) node[inner sep=1pt,fill,circle]{};
\draw(3.8,-0.3) node{$\eta_1$};
\draw(4, 0) node[inner sep=1pt,fill,circle]{};
\draw(4.5, 0) node[inner sep=1pt,fill,circle]{};
\draw(4.5, -0.3) node{$\overline X_1$};
\end{tikzpicture}
\caption{$\text{p-value} <\alpha_1$}
\end{subfigure} \hspace{0.2cm}
\begin{subfigure}{.45\textwidth}
\centering
\begin{tikzpicture}[scale=0.9, every node/.style={scale=0.9}]
\draw[draw = white,fill=gray!20!white]  (1, 2.9)--(1, 3.5)--(4,3.5)--(4,2)--cycle;
\draw[draw = white,fill=gray!20!white]  (4,-0.5)--(4,3.5)--(5,3.5)--(5,-0.5)--cycle;
\draw[dotted, line width=0.3mm] (2.5,0)--(2.5,2.45);
\draw[dotted, line width=0.3mm] (2.5,2.45)--(0,2.45);
\draw[thick,->] (-0.5,0)--(5,0) node[right]{$x_1$};
\draw[thick,->] (0,-0.5)--(0,3.5) node[above]{$x_2$};   
\draw[domain=-0.5:5,dotted, line width=0.5mm] plot (\x,{3.2-0.3*\x});
\draw[domain=-0.5:5,dotted, line width=0.5mm] plot (\x,{2-0.3*\x});
\draw[dotted, line width=0.3mm] (1,-0.5)--(1,3.5);
\draw[dotted, line width=0.3mm] (4,-0.5)--(4,3.5);
\draw(2.1,2) node{$\frac{n_1+n_2}{n_2}\eta_2-\frac{n_1}{n_2}x_1$};
\draw(2.1,3.2) node{$\frac{n_1+n_2}{n_2}\overline X -\frac{n_1}{n_2}x_1$};
\draw(0.8,-0.3) node{$\eta_0$};
\draw(1, 0) node[inner sep=1pt,fill,circle]{};
\draw(3.8,-0.3) node{$\eta_1$};
\draw(4, 0) node[inner sep=1pt,fill,circle]{};
\draw(2.5,0) node[inner sep=1pt,fill,circle]{};
\draw(2.5,-0.3) node{$\overline X_1$};
\draw(0, 2.45) node[inner sep=1pt,fill,circle]{};
\draw(-0.3, 2.45) node{$\overline X_2$};
\end{tikzpicture}
\caption{$\text{p-value} \in [\alpha_1,1-\alpha_0]$}
\end{subfigure} 
\begin{subfigure}{.45\textwidth}
\centering
\begin{tikzpicture}[scale=0.9, every node/.style={scale=0.9}]
\draw[draw = white,fill=gray!20!white]  (0.3,-0.5)--(0.3,3.5)--(5,3.5)--(5,-0.5)--cycle;
\draw[thick,->] (-0.5,0)--(5,0) node[right]{$x_1$};
\draw[thick,->] (0,-0.5)--(0,3.5) node[above]{$x_2$};   
\draw[domain=-0.5:5,dotted, line width=0.5mm] plot (\x,{2-0.3*\x});
\draw[dotted, line width=0.3mm] (1,-0.5)--(1,3.5);
\draw[dotted, line width=0.3mm] (4,-0.5)--(4,3.5);
\draw(2.1,2) node{$\frac{n_1+n_2}{n_2}\eta_2-\frac{n_1}{n_2}x_1$};
\draw(0.8,-0.3) node{$\eta_0$};
\draw(1, 0) node[inner sep=1pt,fill,circle]{};
\draw(3.8,-0.3) node{$\eta_1$};
\draw(4, 0) node[inner sep=1pt,fill,circle]{};
\draw(0.3, 0) node[inner sep=1pt,fill,circle]{};
\draw(0.3, -0.3) node{$\overline X_1$};
\end{tikzpicture}
\caption{$\text{p-value} >1-\alpha_0$}
\end{subfigure} 
\caption{\hspace{0.2cm}Computation of the p-value for the trial performed in two stages.}
\label{rejection_region:pvalue}
\end{figure}

In terms of p-values, this design becomes:

\begin{equation}\tag{D2}
\begin{aligned}
&\text{Reject } \mathcal H_0 \text{ if } \text{p-value}\left(\overline X_1, \overline X_2\right)<\alpha, \text{ where}\\
& \text{p-value} \left(\overline X_1, \overline X_2\right)= \begin{cases}
1- \Phi\left(\overline{X_1} \sqrt{\frac{n_1}{2}}\right), \hspace{4cm} \overline X_1 \notin [\eta_0,\eta_1], \\ 
\int\limits_{\eta_0}^{\eta_1} \sqrt{\frac{n_1}{2}} \,\varphi\left( x_1 \sqrt{\frac{n_1}{2}} \right)\left(1 - \Phi\left(\frac{\sqrt{n_2}\overline X_2+\frac{n_1}{\sqrt{n_2}}\left(\overline X_1-x_1\right)}{\sqrt 2}\right)\right) dx_1 +\alpha_1 , \\ \hspace{7cm}\overline X_1 \in [\eta_0,\eta_1].
\end{cases}
\end{aligned}
\end{equation}

\par The computation of the p-value is complicated by the fact that the grey region in Fig.\ref{rejection_region:pvalue}(b) is not simple. The p-value is an integral of the joint density function $f_{\left(\overline{X_1},\overline{X_2}\right)}(x_1, x_2)$ under $\mathcal H_0$ over the grey area in Fig.\ref{rejection_region:pvalue}. If $\overline X_1 \in [\eta_0,\eta_1]$, we integrate above the line $\frac{n_1+n_2}{n_2}\overline{X}-\frac{n_1}{n_2}x_1 = \overline{X_2} + \frac{n_1}{n_2}\left(\overline{X_1}-x_1\right)$ and for the fixed $x_1$ the probability that $x_2$ is above this line equals $1 - \Phi\left(\frac{\sqrt{n_2}\overline X_2+\frac{n_1}{\sqrt{n_2}}\left(\overline X_1-x_1\right)}{\sqrt 2}\right)$.
Under $\mathcal H_0$ this p-value is uniformly distributed on the interval $(0,1)$, because for any two different realizations of $\left(\overline X_1, \overline X_2\right)$, one of the grey areas (see Fig.\ref{rejection_region:pvalue}) contains the other one.

\par The design $D$ is fully determined by the sample sizes in the first and second stages, $n_1$ and $n_2$, and the probabilities $\alpha_0$ and $\alpha_1$; the overall level of the procedure is $\alpha$. These values often satisfy $1-\alpha_0>\alpha>\alpha_1$. If this does not hold, the second stage is not needed. The thresholds $\eta_0, \, \eta_1, \,\eta_2$ can be computed as 
\begin{align}
&\eta_0 = z_{\alpha_0} \sqrt{\frac{2}{n_1}}, \hspace{0.5cm} \eta_1 = z_{1-\alpha_1}\sqrt{\frac{2}{n_1}},\nonumber\\
&\eta_2 \text{ such that } \nonumber \\
& \int_{\eta_0}^{\eta_1}\sqrt{\frac{n_1}{2}}\,\varphi\left(x_1\,\sqrt{\frac{n_1}{2}}\right)\left(1 - \Phi\left(\left(\frac{n_1+n_2}{n_2}\eta_2-\frac{n_1}{n_2}x_1\right)\sqrt{\frac{n_2}{2}}\right)\right) dx_1 +\alpha_1= \alpha.
\end{align}

Denote the maximum type II error rate for the two-stage design as $\beta_2^{se}(D, \alpha)$. To compute it, we use an expression for the probability of the false negative:

\begin{align}
\label{eq:true_beta}
&\beta_2(D,\alpha,\mu,p)=\mathbb P\left(\text{p-value}\left(\overline X_1,\overline X_2\right) \ge \alpha|\mu, p\right)\nonumber\\
&= \begin{cases}
\beta\left(n_1,z_{1-\alpha_1}\sqrt{\frac{2}{n_1}},\mu,p\right), \hspace{5cm}\alpha\notin(\alpha_1,1-\alpha_0), \\
\beta(n_1,\eta_0,\mu,p)+\int^{\eta_1}_{\eta_0} \beta_\eta(n_1,x_1,\mu,p)\beta\left(n_2,\frac{n_1+n_2}{n_2}\eta_2(D, \alpha)-\frac{n_1}{n_2}x_1,\mu,p\right)dx_1,\\
\hspace{9cm} \alpha\in(\alpha_1,1-\alpha_0).
\end{cases}
\end{align}
We introduce the case $\alpha\notin(\alpha_1,1-\alpha_0)$ in \eqref{eq:true_beta}, because we will need it later for the multicenter two-stage design.
Similar to the one-stage trial, the following holds.
\begin{lemma} For the region of strong effect defined as in Definition \ref{def:se}, the maximum type II error rate for the two-stage trial is
$$\beta_2^{se}(D,\alpha)=\underset{i=1, ..., s}{\max}\beta_2(D,\alpha,\mu_i,p_i).$$
\end{lemma}

\subsection{Multicenter RCT designs}\label{many_groups_algorithm}
\par In this section, we generalize the designs proposed above to multicenter RCTs. We suggest first to conduct single center RCTs with identical designs in all centers in parallel, with each center having its own control group. Following this initial stage, the results are analyzed together in order to make conclusions about the existence of the subpopulation in each center. \par Combining the evidence across multiple centers requires p-values. For this reason, in the previous subsection, we introduced an equivalent designs formulation using p-values. We assume patient enrollment in all centres is done independently, and that all the p-values are independent. 
\par In the single center design one rejects $\mathcal H_0$ if the $\text{p-value}$ is below $\alpha$, but to control the family wise type I error rate in the multi-center context at the level $\alpha$, we need a multiple testing procedure, such as Hochberg's step-up procedure ~\cite{HOCHBERG:1988aa}. The whole design is then described as follows:
\begin{align*}
\bullet & \text{ Let } \text{p-value}_1,\text{p-value}_2,\ldots,\text{p-value}_M \text{ be the set of p-values}\\ &\text{ at the end of each study (see Eq.D2). Put the p-values in an } \\
& \text{ ascending order: }\text{ p-value}(1) \le \text{p-value}(2)\le\cdots\le\text{p-value}(M). \\
\bullet & \text{ Denote the index of the center corresponding to } \text{p-value}(k) \text{ as } c_{(k)}. \tag{D3} \\ 
\bullet & \text{ Given the thresholds, } \alpha(1)\le\alpha(2)\le\cdots\le\alpha(M), \text{ find the index } \\ & \text{K corresponding to the largest } k \text{ such that } \text{ p-value}(k)\le \alpha(k). \\ & \text{ In Hochberg's step-up procedure } \alpha(k) = \frac{\alpha}{M+1-k}. \\
\bullet &\text{ The centers } c_{(1)} \text{ through } c_{(K)} \text{ are then judged to have sensitive} \\ & \text{ subpopulations, and the null hypothesis is rejected for them.}
\end{align*}
There is a complication for a two-stage multicenter design.
\begin{figure}%
\begin{subfigure}{0.45\textwidth}
 \centering
 \resizebox{\textwidth}{0.9\height}{
\begin{tikzpicture}
\draw[thick,-] (0,0)--(7,0);
\draw(0, -0.4) node{$0$};
\draw(0.6, 1.8) node{$\alpha(1)$};
\draw(1.6, 1.8) node{$\alpha(2)$};
\draw(2.3, 2.8) node{$\alpha_1$};
\draw(2.6, 1.8) node{$\alpha(3)$};
\draw(3.6, 1.8) node{$\alpha(4)$};
\draw(3.8, 2.8) node{$\alpha$};
\draw(5.8, 2.8) node{$1-\alpha_0$};
\draw(7, -0.4) node{$1$};
\draw (0, 2pt) -- (0, -2pt);
\draw (7, 2pt) -- (7, -2pt);
\draw (1, 3.5) -- (1, -0.2);
\draw (2, 3.5) -- (2, -0.2);
\draw (3, 3.5) -- (3, -0.2);
\draw (4, 3.5) -- (4, -0.2);
\draw [dotted] (2.5,3.5)--(2.5, -0.2);
 \draw [dotted] (6,3.5)--(6, -0.2);
\draw(1, 1.5) node[inner sep=2pt,fill,circle]{};
\draw(2, 1.5) node[inner sep=2pt,fill,circle]{};
\draw(4, 2.5) node[inner sep=2pt,fill,diamond]{};
\draw(2.5, 2.5) node[inner sep=2pt,fill,diamond]{};
\draw(3, 1.5) node[inner sep=2pt,fill,circle]{};
\draw(4, 1.5) node[inner sep=2pt,fill,circle]{};
\draw(6, 2.5) node[inner sep=2pt,fill,diamond]{};
\draw(0.5, 0) node[inner sep=3pt,fill = black!30!green]{};
\draw(2.2, 0) node[inner sep=3pt,fill = black!30!green]{};
\draw(2.8, 0) node[inner sep=3pt,fill= black!30!green]{};
\draw(4.5, 0) node[inner sep=3pt,fill= black!30!red]{};
\draw(0.5, -0.4) node{$c_{(1)}$};
\draw(2.2, -0.4) node{$c_{(2)}$};
\draw(2.8, -0.4) node{$c_{(3)}$};
\draw(4.5, -0.4) node{$c_{(4)}$};
\end{tikzpicture}}
\end{subfigure}\hspace{0.5cm}
\begin{subfigure}{0.45\textwidth}
 \centering
 \resizebox{\textwidth}{0.9\height}{
\begin{tikzpicture}
\draw[thick,-] (0,0)--(7,0);
\draw(0, -0.4) node{$0$};
\draw(0.6, 1.8) node{$\alpha(1)$};
\draw(1.6, 1.8) node{$\alpha(2)$};
\draw(2.3, 2.8) node{$\alpha_1$};
\draw(2.6, 1.8) node{$\alpha(3)$};
\draw(3.6, 1.8) node{$\alpha(4)$};
\draw(3.8, 2.8) node{$\alpha$};
\draw(5.8, 2.8) node{$1-\alpha_0$};
\draw(7, -0.4) node{$1$};
\draw (0, 2pt) -- (0, -2pt);
\draw (7, 2pt) -- (7, -2pt);
\draw (1, 3.5) -- (1, -0.2);
\draw (2, 3.5) -- (2, -0.2);
\draw (3, 3.5) -- (3, -0.2);
\draw (4, 3.5) -- (4, -0.2);
\draw [dotted] (2.5,3.5)--(2.5, -0.2);
 \draw [dotted] (6,3.5)--(6, -0.2);
\draw(1, 1.5) node[inner sep=2pt,fill,circle]{};
\draw(2, 1.5) node[inner sep=2pt,fill,circle]{};
\draw(4, 2.5) node[inner sep=2pt,fill,diamond]{};
\draw(2.5, 2.5) node[inner sep=2pt,fill,diamond]{};
\draw(3, 1.5) node[inner sep=2pt,fill,circle]{};
\draw(4, 1.5) node[inner sep=2pt,fill,circle]{};
\draw(6, 2.5) node[inner sep=2pt,fill,diamond]{};
\draw(0.5, 0) node[inner sep=3pt,fill = black!30!green]{};
\draw(2.2, 0) node[inner sep=3pt,fill = black!30!red]{};
\draw(3.5, 0) node[inner sep=3pt,fill= black!30!red]{};
\draw(4.5, 0) node[inner sep=3pt,fill= black!30!red]{};
\draw(0.5, -0.4) node{$c_{(1)}$};
\draw(2.2, -0.4) node{$c_{(2)}$};
\draw(3.5, -0.4) node{$c_{(3)}$};
\draw(4.5, -0.4) node{$c_{(4)}$};
\end{tikzpicture}}
\end{subfigure}\vspace{0.5cm}
\caption{\hspace{0.2cm}The figure shows two examples. The green squares corresponds to the centers for which $\mathcal H_0$ was rejected and the red ones to the centers where it was not rejected. In both $c_{(1)}$ and $c_{(2)}$ the study was terminated after the first stage (their p-values are less than $\alpha_1$). However, the decision to reject the null for $c_{(2)}$ depends on the results in the centers $c_{(3)}$ and $c_{(4)}$. On the left-hand side, the null is rejected at $c_{(2)}$ and $c_{(3)}$. On the right-hand side, decisions about $c_{(2)}, c_{(3)}, c_{(4)}$ are all non-rejections.}
\label{allocation_multicenters}
\end{figure}
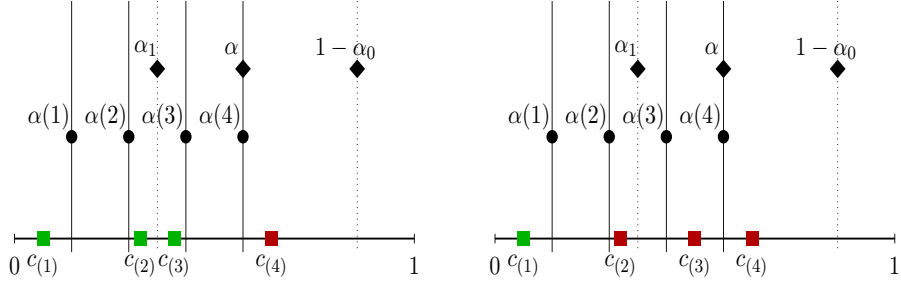
In the case of a single center, if the $\text{p-value}<\alpha_1$, then $\mathcal H_0$ is rejected after the first stage. When more than one center is involved, the condition $\text{p-value}_i<\alpha_1$, which only occurs when we stop after the first stage, does not necessarily mean that for this center the null is rejected. If the design parameter $\alpha_1$ is such that $\alpha(1)>\alpha_1$, the rejection always occurs. Otherwise, the rejection depends on the other p-values or, more precisely, on their arrangement with regard to $\alpha(k)$ which is only known at the completion of the second stage.

Consider two p-value realisations shown in Fig.\ref{allocation_multicenters}. There are two centers, $c_{(1)}, c_{(2)}$ where the trial was stopped after the first stage, and their p-values are less than $\alpha_1$. In two other centers the second stage was conducted and their p-values are in the interval $(\alpha_1, 1-\alpha_0)$. For $c_{(1)}$ the null is rejected after the first stage in both situations, because $\text{p-value}_1<\alpha(1)$, but the decision about $c_{(2)}$ depends on the results in the centers where the second stage is conducted. To summarize, if for some center after the first stage $\alpha(k)<\text{p-value}(k)<\alpha_1$, the final decision about this center might depend on the results from the centers where the trial continues to the second stage.  
 
Recall that the type II error is defined as "at least $m$ false nulls among $M_1$ centers with strong effect are not rejected, $m\le M_1\le M$". Denote the maximum type II error rates as $\beta_{\text{fw}}^{se}(M_1, m)$. It is complicated to calculate $\beta_{\text{fw}}^{se}(M_1, m)$ exactly for arbitrary couples $(M_1, \, m)$, but the upper bound can be obtained rather straightforwardly. To do this, we introduce auxiliary values
\begin{align}
\beta^{se}_j = \begin{cases} \beta^{se}(n, \alpha(j)), &\text{one-stage design},\\
\beta^{se}_2(D, \alpha(j)), &\text{two-stage design}.\end{cases}
\label{ind:errors}
\end{align}

\begin{lemma}\label{lemma:multi_beta_max}The following inequality always holds:
\begin{equation}
1-\beta^{se}_{\text{fw}}(M_1, m) \ge \left(1-\beta_{M_1+1-m}^{se}\right)^{M_1+1-m}.
\end{equation}
If $M_1=M$ and $m=1$, equality is achieved.
\end{lemma}
From this lemma, we conclude that $\beta^{se}_{\text{fw}}(M_1, m)\le 1-\left(1-\beta_{M_1+1-m}^{se}\right)^{M_1+1-m}$ and $\beta^{se}_{\text{fw}}(M, 1)=1-\left(1-\beta_{M}^{se}\right)^{M}$.

\section{Planning the trial}\label{Sec:Planning}
In this section we present a framework which helps the user to choose the design parameters in order to control the type II error rate below some given level $\beta_{\max}$ or the power above $(1-\beta_{\max})$. For the multi-centre trial we suggest to control $\beta^{se}_{\text{fw}}(M,1)$. Here we assume $\beta_{\max} < 0.5$ for the single center RCT and $\beta_{\max} < 1- 0.5^M$ for the multicenter RCT. At the beginning the user should carefully define the region of strong effect as in Definition \ref{def:se}. Next, the user chooses $\beta_{\max}$ and the design parameters. We also provide additional information that may influence the user's choice of $\beta_{\max}$. In the case of changing $\beta_{\max}$, one should recalculate the design parameters. To illustrate the planning procedure, we use the region of strong effect as depicted in Fig.~\ref{fig:n_1stage} and $\beta_{\max} = 0.2$.  

\subsection{Planning a single center RCT}\label{App1:approximation_clt}
\subsubsection{One-stage design: choosing the sample size}
The planning of a single center trial proceeds by choosing $\alpha$ and computing the sample size that ensures the desired minimal power. After $\alpha$ is chosen (here and elsewhere $\alpha = 0.05$), $n$ is computed as
\begin{align}
n =  \min\left( n \,:\, \beta^{se}(n, \alpha) \le \beta_{\max}\right).
\label{eq:true_n}
\end{align}
To compute $\beta^{se}(n,\alpha)$, we use an asymptotic approximation of $\beta(n, \eta, \mu,p)$ based on the central limit theorem. Notice that $Y^T_1, \ldots, Y^T_n$ are i.i.d. random variables with $\mathbb E\left(Y_i^T\right) = \mu p$ and $\Var\left(Y^T_i\right) =  1+(1-p)p\mu^2$. Hence, 
$$\overline{Y^T} \sim \mathcal N\left(\mu p, \frac{1+(1-p)p\mu^2}{n}\right), \hspace{0.2cm}\overline X \sim \mathcal N\left(\mu p, \frac{1+(1-p)p\mu^2}{n}+\frac{1}{n}\right).$$
With Barry--Essen's upper bound (~\cite{Berry:1941aa, essen}) on the absolute deviation of the probability of a false negative from its normal approximation, we have
\begin{equation}\label{eq:beta_approx}
\begin{aligned}
& \Biggr\rvert \beta(n,\eta,\mu,p) - \Phi\left(\frac{ \sqrt{n}(\eta-\mu p)}{\sqrt{2+(1-p)p \mu^2}}\right)\Biggr\rvert \le \frac{C}{\sqrt n}.
\end{aligned}
\end{equation}
We use this normal approximation in the following computations, because the real difference in \eqref{eq:beta_approx} is negligible even for small sample sizes. From Lemma \ref{lemma:beta_max}, 
$$\beta^{se}(n,\alpha) \approx \underset{i=1,\ldots,s}{\max}\Phi\left(\frac{\sqrt 2z_{1-\alpha} -\sqrt{n}\mu_i p_i}{\sqrt{2+(1-p_i)p_i \mu_i^2}}\right)$$
and 
\begin{align}\label{eq:n_min_app1}
n \approx \underset{i = 1, \ldots, s}{\max}\left(\frac{\sqrt 2 z_{1-\alpha}+z_{1-\beta_{\max}}\sqrt{2+(1-p_i)p_i \mu^2_i}}{\mu_i p_i}\right)^2.
\end{align}
\subsubsection*{Example} Let the region $\cal E$ be as in Fig.~\ref{fig:n_1stage} and let the maximum type II error rate be $0.2$. The minimum sample size for the trial with these parameters from \eqref{eq:n_min_app1} is $n \approx 85.3$, whereas the exact value from \eqref{eq:true_n} is $86$. For $n=86$ and $\alpha = 0.05$, the rejection threshold $\eta = 0.251$ and it is straightforward to conduct the trial (D1).


\begin{figure}
\centering
\begin{tikzpicture}[scale = 4]
\draw[thick,->] (0,0)--(2.2,0) node[right]{$\mu$};
\draw[thick,->] (0,0)--(0,1.1) node[above]{$p$};
\draw[draw = white,fill=gray!20!white] (0.7,1)--(0.7, 0.6) -- (1,0.6) -- (1, 0.4)--(2, 0.4)--(2,0.2) -- (2.2, 0.2)--(2.2, 1) --cycle;
\draw[line width=0.3mm] (0.7,1)--(0.7, 0.6) -- (1,0.6) -- (1, 0.4)--(2, 0.4)--(2,0.2) -- (2.2, 0.2);
\draw(1.5, 0.7) node{${\cal E}$};
 \foreach \x/\xtext in {0, 0.2, 0.4, 0.6, 0.8, 1, 1.2, 1.4, 1.6, 1.8, 2} 
   \draw (\x cm,0.5pt) -- (\x cm,-0.5pt) node[anchor=north,fill=white] {$\xtext$};
 \foreach \y/\ytext in {0, 0.2, 0.4, 0.6, 0.8, 1} 
   \draw (0.5pt,\y cm) -- (-0.5pt,\y cm) node[anchor=east,fill=white] {$\ytext$};
\end{tikzpicture}
\caption{\hspace{0.2cm}Region of strong effect, characterised by $\vec \mu = (2, 1, 0.7),\,  \vec p = (0.2,0.4,0.6)$.}
\label{fig:n_1stage}
\end{figure}
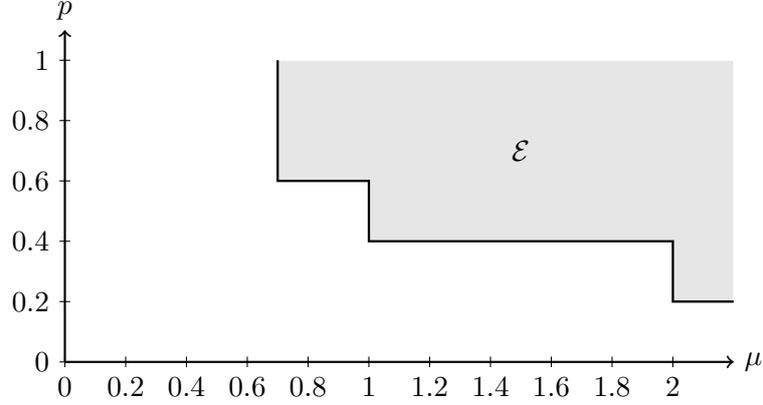

\subsubsection{Two-stage design}
The design is determined by $D = \{n_1, \, \alpha_0, \, \alpha_1, \, n_2\}$. After $\alpha$ is chosen, the sample size at the second stage is computed as
\begin{align}\label{eq:n2_min^app}
n_2 =  \min\left(n\, : \, \beta^{se}_2(D, \alpha) \le \beta_{\max}\right).
\end{align}
To approximate $n_2$, we use the following approximation for the type two error for the two-stage design:
\begin{align*}
&\beta_2(D, \alpha, \mu, p)\\
& \approx \Phi\left(\frac{\sqrt{n_1}(\eta_0 - \mu p)}{\sqrt{2+(1-p)p\mu^2}}\right) + \int_{\eta_0}^{\eta_1}\frac{\varphi\left(\frac{\sqrt n_1(x_1-\mu p)}{\sqrt{2+(1-p)p\mu^2}}\right)}{\sqrt{\frac{2+(1-p)p\mu^2}{n_1}}}\Phi\left(\frac{\frac{n_1+n_2}{n_2}\eta_2(D, \alpha) - \frac{n_1}{n_2}x_1-\mu p}{\sqrt{\frac{2+(1-p)p\mu^2}{n_2}}}\right)dx_1.
\end{align*}
We next compute $n_2$ numerically. Note that it is possible that no design fulfills the condition on $ \beta_{\max}$. In the following we will address this problem.

For planning purposes we next consider the expected sample size under the null,
\begin{equation}\label{q0}
q_0(D)=n_1+(1-\alpha_0-\alpha_1)n_2
\end{equation}
and the analogous formula for the maximum expected sample size under the alternative is
\begin{equation}\label{q1max}
q_1(D)=n_1+\underset{\mu>0, \, p>0}{\max}\left(\beta\left(n_1,\frac{\sqrt 2\, z_{1-\alpha_1}}{\sqrt{n_1}},\mu,p\right)-\beta\left(n_1,\frac{\sqrt 2\,z_{\alpha_0}}{\sqrt{n_1}},\mu,p\right)\right)n_2.
\end{equation}
In this paper we put $\alpha_0>0.5$ in order to stop the trial early in the case of the absence of the effect. With this choice, both $z_{1-\alpha_1}$ and $z_{\alpha_0}$ are positive.
\begin{lemma}\label{maxq1}
If $0<z_{\alpha_0}<z_{1-\alpha_1}$, the maximum probability of conducting the second stage in \eqref{q1max} approximately satisfies
\begin{equation}
\underset{\mu>0, \, p>0}{\max}\left(\beta\left(n_1,\frac{\sqrt 2\, z_{1-\alpha_1}}{\sqrt{n_1}},\mu,p\right)-\beta\left(n_1,\frac{\sqrt 2\,z_{\alpha_0}}{\sqrt{n_1}},\mu,p\right)\right) \approx  2\Phi\left(\frac{z_{1-\alpha_1}-z_{\alpha_0}}{2}\right)-1.
\end{equation}
The maximum is achieved for $\mu = \frac{z_{1-\alpha_1}+z_{\alpha_0}}{\sqrt{2 n_1}},\,p=1$.
\end{lemma}
From Lemma \ref{maxq1}, it follows that $q_1(D) \approx n_1 + \left(2\Phi\left(\frac{z_{1-\alpha_1}-z_{\alpha_0}}{2}\right)-1\right)n_2$. 
\par We let the user choose $n_1$ and $\alpha_0$, but we suggest setting $\alpha_1$ such that it minimizes $q_1(D)$ for given $n_1,\,\alpha_0$. This leads to the choice $\alpha_1  =  \underset{\alpha_1}{\argmin\,} q_1(D)$.
Once $n_1,\, \alpha_0,\, \alpha_1$ have been fixed, all the other quantities including $n_2$, $q_0$, $q_1$ can be computed. 

In the following we discuss guidance for the user on the choice of $n_1$ and $\alpha_0$.
We suppose that the choice of $n_1, \,\alpha_0$ is based on $q_{0}, \, q_{1}$ and  $n_2$. We note that the design does not exist for every pair $(n_1, \alpha_0)$. With $0.5<\alpha_0<1-\alpha$ and  \eqref{eq:true_beta}, we have $\beta^{se}(n_1, 1-\alpha_0) < \beta^{se}_2(D, \alpha) \le \beta_{\max}$, and with an approximation for $\beta^{se}(n_1, 1-\alpha_0)$, mentioned above,
\begin{equation}\label{eq:alpha0_const}
\alpha_0 < \underset{i=1, ..., s}{\min}{\Phi\left(\sqrt{\frac{n_1}{2}}\, \mu_i p_i-z_{1-\beta_{\max}}\sqrt{1+\frac{(1-p_i)p_i}{2}\mu^2_i}\right)}.
\end{equation}
This in turn implies that
\begin{align}\label{eq:n1_const}
n_1>\left( \underset{i=1, ..., s}{\max}\frac{z_{1-\beta_{\max}}\sqrt{2+(1-p_i)p_i \mu^2_i}}{\mu_i p_i}\right)^2.
\end{align}

\begin{figure}
\centering
\includegraphics[width = 10cm]{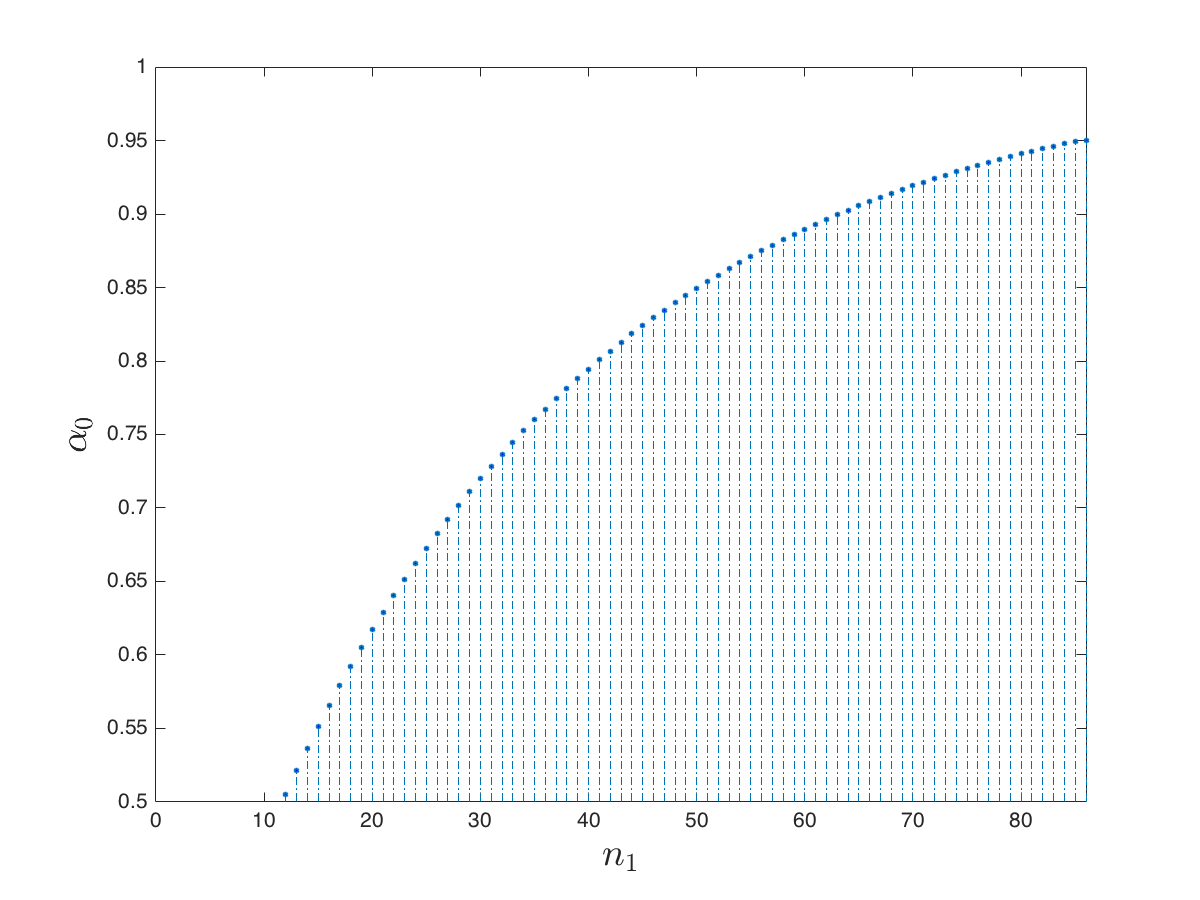}
\caption{\hspace{0.2cm}Possible pairs of $(n_1, \alpha_0)$ for the two-stage design. The parameters are $\alpha = 0.05, \, \beta_{\max} = 0.2, \, n = 86$; \hspace{0.1cm} $12\le n_1\le 86$.}
\label{fig:choice_alpha_n1}
\end{figure}


\begin{figure} 
\begin{center}
\begin{subfigure}{.6\textwidth}
\includegraphics[width = \textwidth]{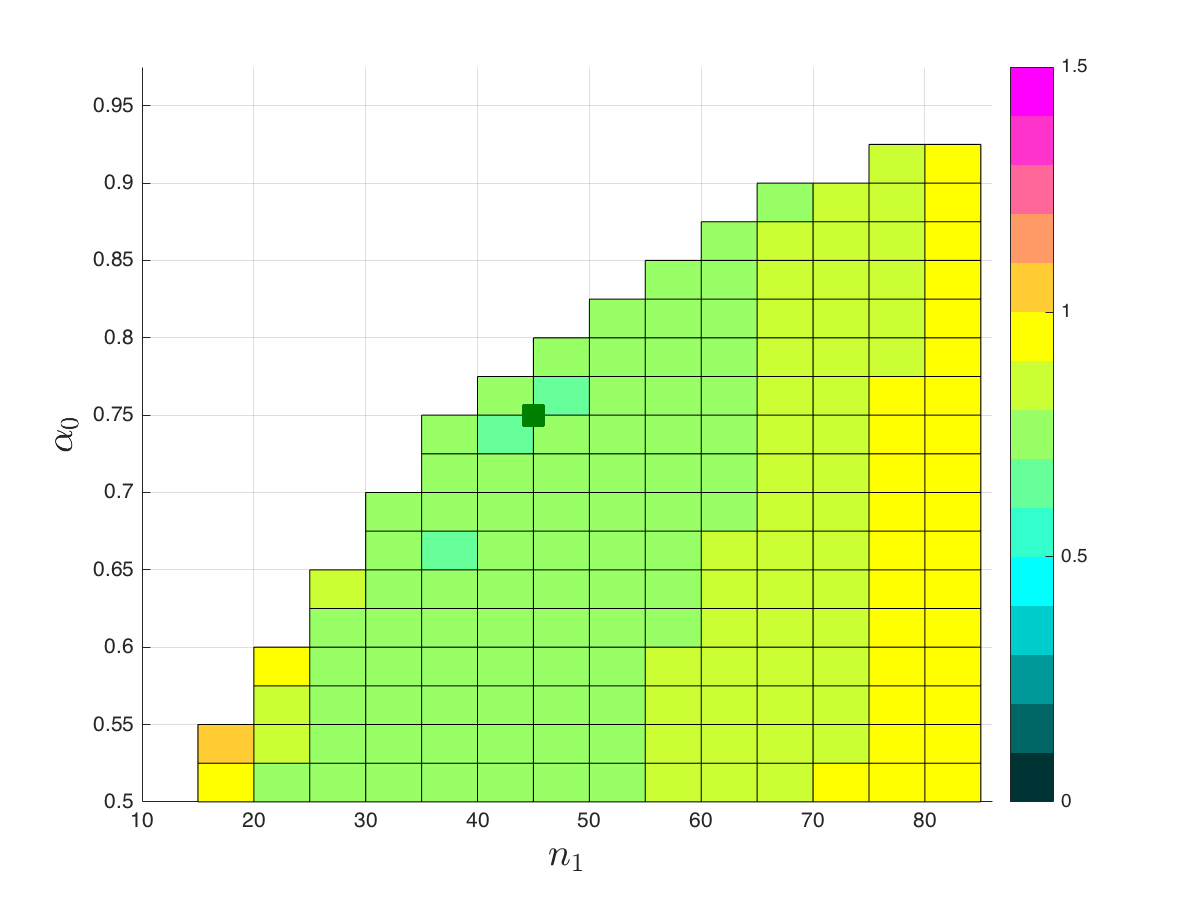}
\caption{$\frac{q_{0}}{n}$. $\min q_{0}  = 60$ for $n_1 = 45, \, \alpha_0 = 0.75$.}
\label{char1}
\end{subfigure}
\begin{subfigure}{.6\textwidth}
\includegraphics[width = \textwidth]{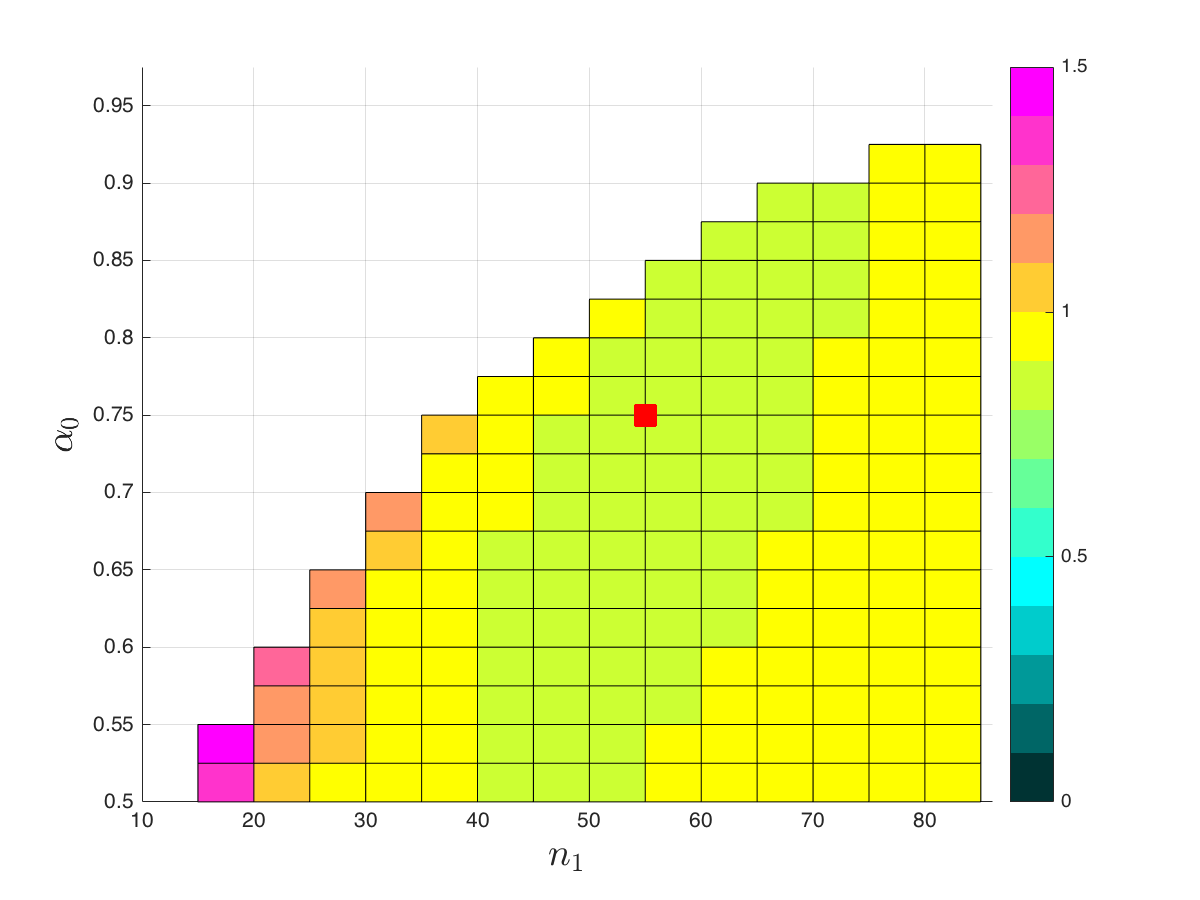}
\caption{$\frac{q_{1}}{n}$. $\min q_{1}  =75$ for $n_1 = 55, \, \alpha_0 = 0.75$}
\label{char2}
\end{subfigure}
\end{center}
\begin{center}
\begin{subfigure}{.6\textwidth}
\includegraphics[width = \textwidth]{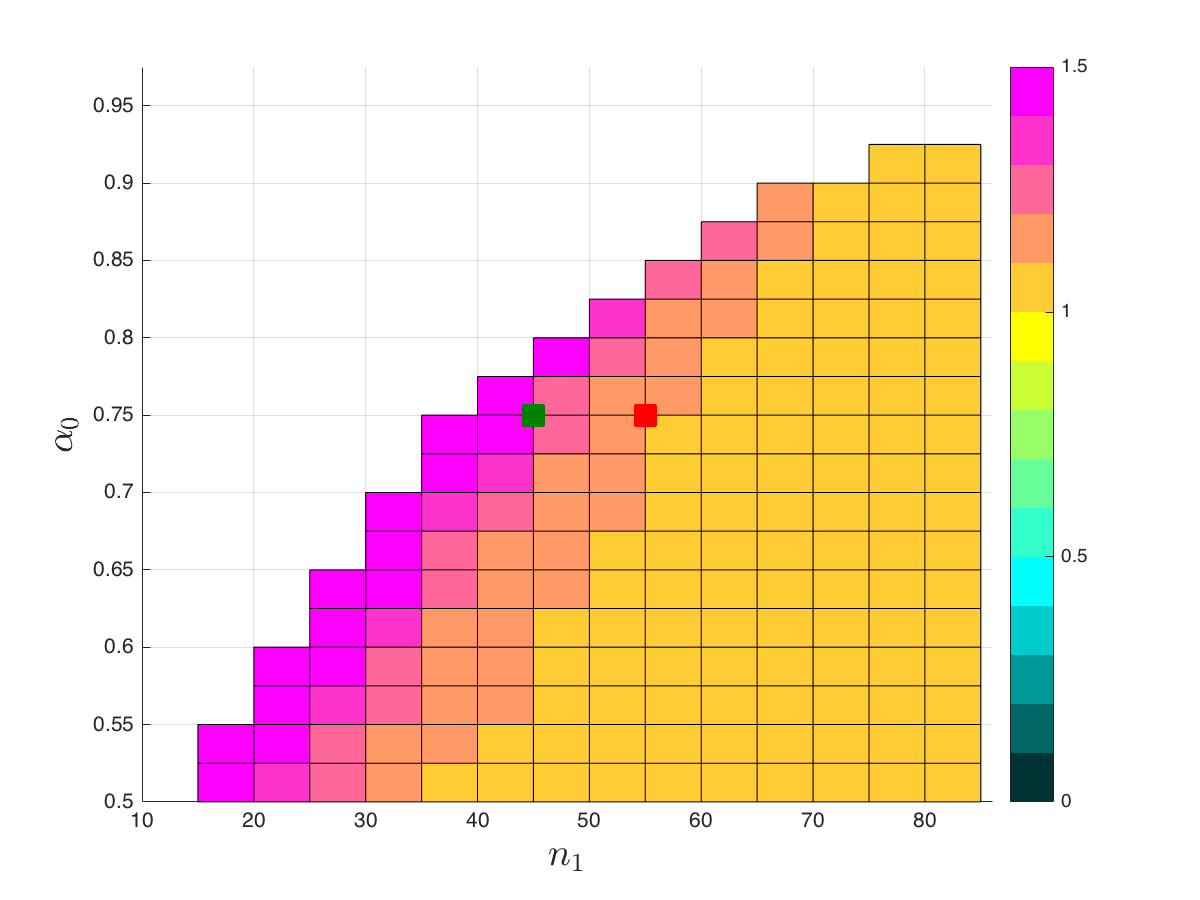}
\caption{$\frac{n_1+n_2}{n}$}
\label{char3}
\end{subfigure}
\end{center}
\caption{\hspace{0.2cm}The diagnostic plot for choosing the design. The parameters are $\alpha = 0.05, \, \beta_{\max} = 0.2, \, \vec \mu= (2, 1, 0.7), \, \vec p = (0.2,0.4,0.6), \, n = 86$. 
The color coding indicates the percentage gain in sample size compared to the one stage design. Small values below one indicate a gain, and values above 1 indicate a loss. The green square in (a) is a pair of parameters that attains the minimum of $q_0$, whereas the red square in (b) attains the minimum for $q_1$. These points are also shown in (c). For the choice $n_1 = 55, \, \alpha_0 = 0.7, \, \alpha_1 = 0.026, \, n_2 = 38, \, \eta_0 = 0.10, \, \eta_1 = 0.37, \, \eta_2 = 0.26$, we have $q_{0} = 66,\, q_{1} = 75, \, n_1+n_2 = 93$.}
\label{choice_n1_alpha0}
\end{figure}

\begin{figure}
\centering
\includegraphics[width =  9cm]{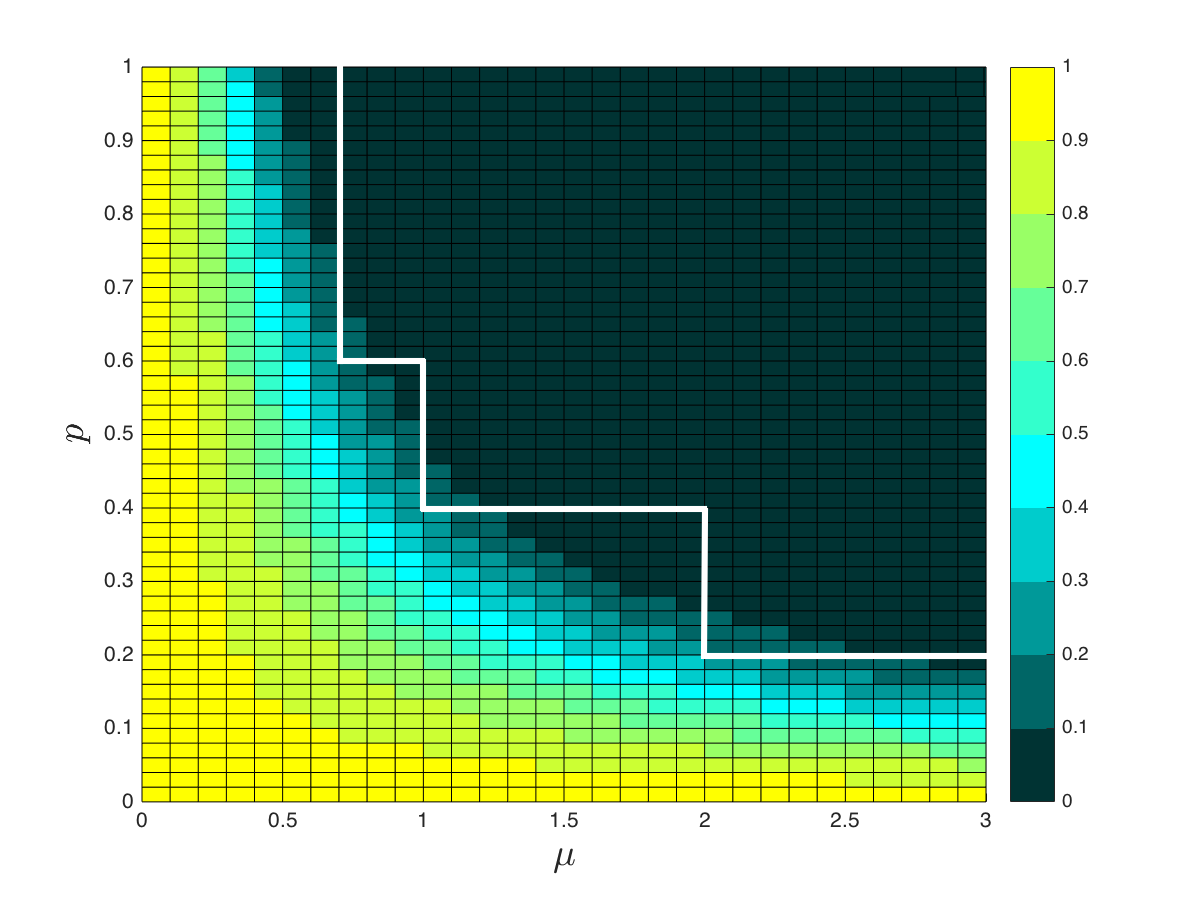}
\caption{\hspace{0.2cm}Probability of a false negative computed for $n_1 = 55, \, \alpha_0 = 0.7, \, \alpha_1 = 0.026, \, n_2 = 38$.}
\label{char5}
\end{figure}
\begin{figure}
\begin{center}
\begin{subfigure}{.7\textwidth}
\includegraphics[width = \textwidth]{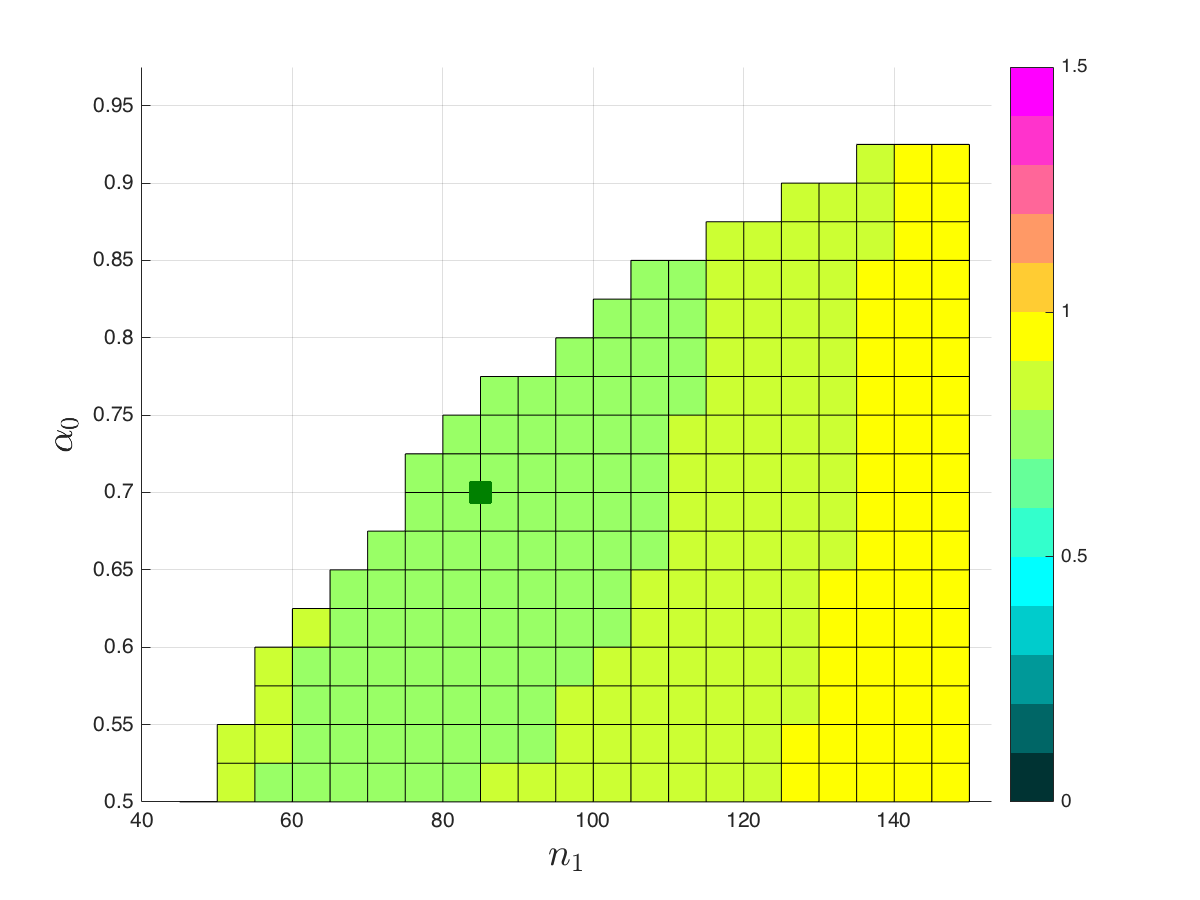}
\caption{$\frac{q_{0}}{n}$. $\min q_{0}  = 113$ for $n_1 = 85, \, \alpha_0 = 0.7$.}
\end{subfigure}
\begin{subfigure}{.7\textwidth}
\includegraphics[width = \textwidth]{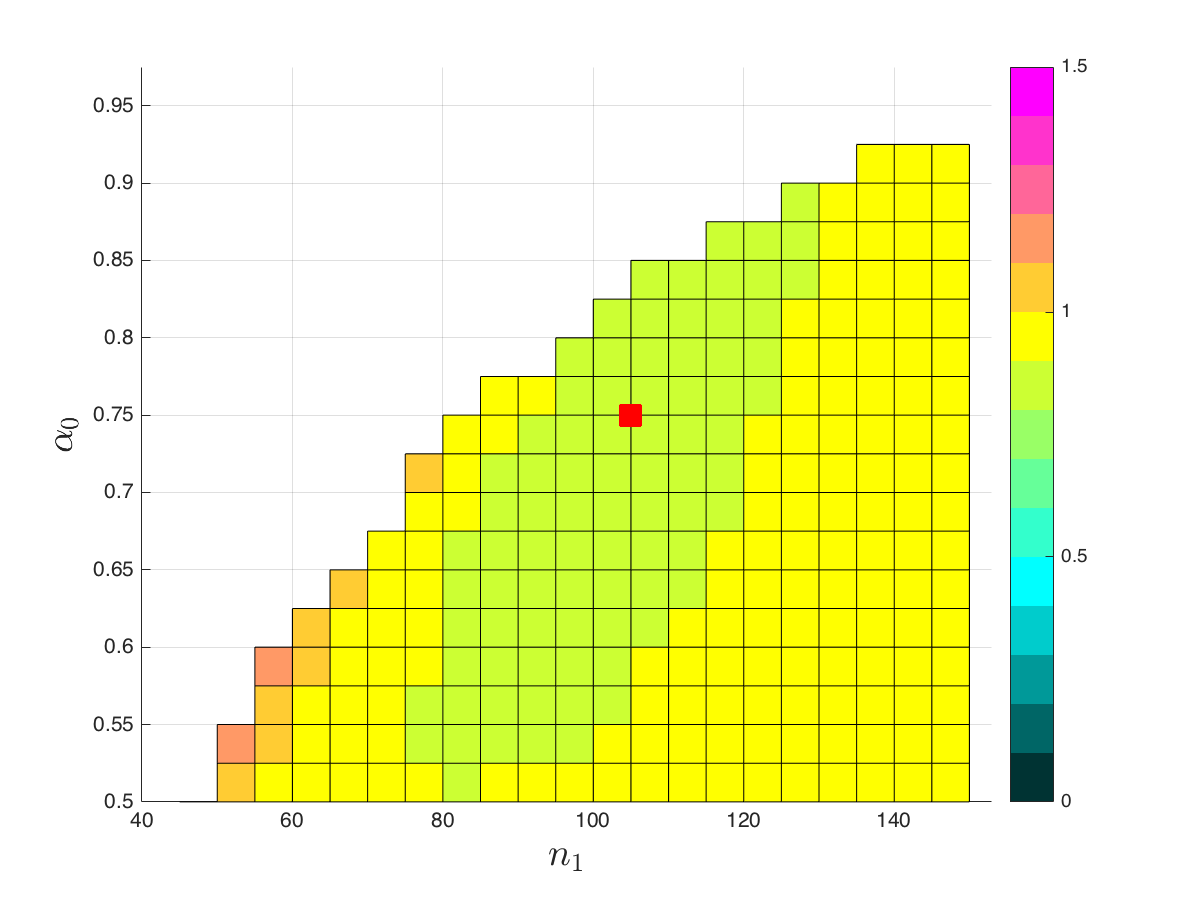}
\caption{$\frac{q_{1}}{n}$. $\min q_{1}  = 134$ for $n_1 = 105, \, \alpha_0 = 0.75$}
\end{subfigure}
\begin{subfigure}{.7\textwidth}
\includegraphics[width = \textwidth]{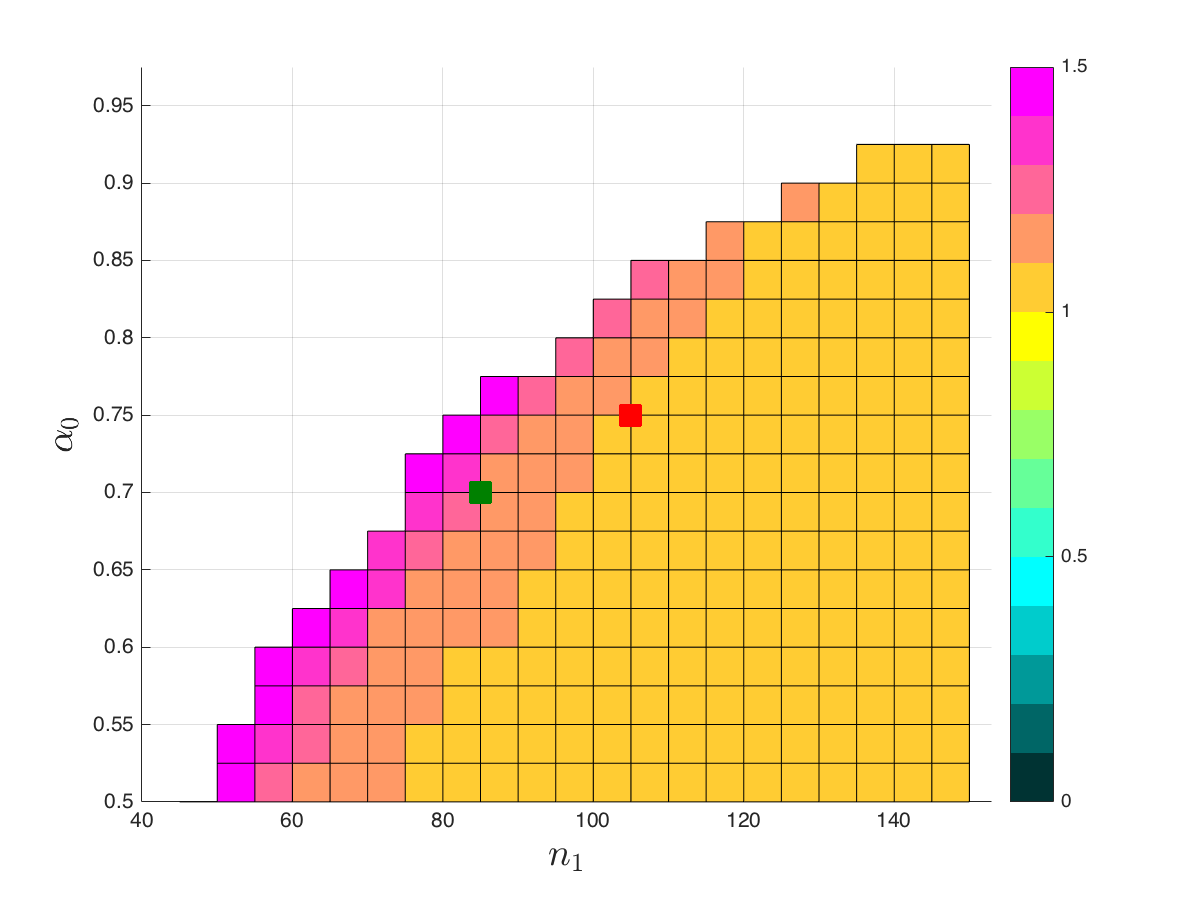}
\caption{$\frac{n_1+n_2}{n}$}
\end{subfigure}
\end{center}
\caption{\hspace{0.2cm}The diagnostic plot for choosing the design. The parameters are $M=4, \,\alpha = 0.05, \,\beta_{\max} = 0.2,\, \vec \mu = (2, 1, 0.7),\,  \vec p = (0.2,0.4,0.6), \, n = 153$. For the parameters $n_1 = 100, \, \alpha_0 = 0.7, \, \alpha_1 = 0.026, \, n_2 = 65 , \, \eta_0 = 0.07, \, \eta_1 = 0.28, \, \eta_2 = 0.19$, we have $q_{0} = 118,\, q_{1} = 134, \, n_1+n_2 = 165$.}
\label{n1alpha0_mult}
\end{figure}

The first stage sample size also should satisfy $n_1<n$, otherwise the two-stage scheme is less efficient than the one-stage design. Based on \eqref{eq:alpha0_const} and \eqref{eq:n1_const}, the region of possible $(n_1, \alpha_0)$ provided to the user is shown in Fig.\ref{fig:choice_alpha_n1}. 
\par Among all possible pairs $(n_1, \alpha_0)$ one would choose ones that make the values of $q_0$, $q_1$ and $n_1+n_{2}$ small. However, these criteria cannot be achieved simultaneously. This trade-off is left for the user to resolve, and in Fig.\ref{choice_n1_alpha0} we provide an example of a graphical help for the choices. Recall that the design is built to ensure the type II error rate to be not greater than $\beta_{\max}$ in all the corners of $\mathcal E$. However, this probability can be very sensitive to $\mu$ and/or $p$. We assume that given the probability of a false negative for all $(\mu,p)$, as shown in Fig.\ref{char5}, the user may decide to change $\beta_{\max}$. If these rates are small, the trial size may be reduced by using a larger $\beta_{\max}$, and if $\beta_{\max}$ is considered to be reduced, then the larger trial is necessary. If he decides to change $\beta_{\max}$, the design needs to elaborated from the beginning.


\subsection{Planning a multicenter RCT}

\par The design is determined by the set $\{M, n_1, \, \alpha_0, \, \alpha_1, \, n_2\}$. For the multicenter trial we build the design to control $\beta^{se}_{\text{fw}}(M,1)$. The type II error rates $\beta^{se}_{\text{fw}}(M_1,m)$ where $M_1 < M$ or $m>1$ are not controlled. From Lemma \ref{lemma:multi_beta_max}, the condition $\beta^{se}_{\text{fw}}(M,1)\le \beta_{\max}$ is equivalent to $\beta^{se}_M \le 1 - (1-\beta_{\max})^{\frac{1}{M}}$. This means that if we consider $M$ centers separately, the design in each of the centers will correspond to the single center RCT with the type I and type II errors controlled at levels $\alpha(M)$, the largest threshold of the multiple rejection procedure, and $\beta^{se}_M = 1 - (1-\beta_{\max})^{\frac{1}{M}}$, respectively. 
\par To illustrate the planning for the RCT with four centers, consider that FWER is controlled at $0.05$ and the family-wise type II error rate is controlled at $\beta_{\max} = 0.2$.  With Hochberg's multiple procedure the design is then determined by $\alpha(4)  = 0.05$ and $\beta^{se}_M = 1 - (1-0.2)^{\frac{1}{4}} \approx 0.054$. As in the single center RCT, the choice of $n_1$ and $\alpha_0$ is made by the user based on the information provided in Fig.\ref{n1alpha0_mult}. 

\par Nonetheless, maximum type II error rates for $M_1< M$ or $m>1$ can be also of interest for the study. Even if an effect exists, it is not necessarily strong in all centers. Furthermore, a trial missing a few centers with strong effect can still be useful if in the other centers it is detected. Hence, given the information about $\beta^{se}_{\text{fw}}(M_1,m)$, the user might change $\beta_{\max}$ and recalculate the design. In Lemma \ref{lemma:multi_beta_max}, we show an upper bound on these errors,
$$ \beta^{se}_{\text{fw}}(M_1,m) \le 1 - \left(1-\beta^{se}_{M_1+1-m}\right)^{M_1+1-m}.$$
Once the parameters $n_1, \,\alpha_0$ have been chosen, the user is supplied with Table \ref{tab:comparison_diffpoints}(a).

\par We do not know an exact value of $\beta^{se}_{\text{fw}}(M_1,m)$, but we can compute the type II error rate $\beta_{\text{fw}}(M_1, m)$ for the given responses (the details can be found in Appendix \ref{Computation of beta fw}). the user has a better intuition about $\beta^{se}_{\text{fw}}(M_1,m)$ if we provide him with $\beta_{\text{fw}}(M_1, m)$ computed for different sets of responses. We take $M_1$ centres with the same $(\mu^*, \, p^*) \in \mathcal E$, while in the remaining $M - M_1$ centres $\mathcal H_0$ is true (see Table \ref{tab:comparison_diffpoints} (b),(c)). We suggest that the user checks more than one pair $(\mu^*, \, p^*)$ to have a better understanding of possible type II error rates. Based on these tables, the user decides whether to change $\beta_{\max}$ or not.

As illustrated, the values of $\beta_{\text{fw}}(M_1, m)$ are very sensitive to the true parameters $(\mu^*, \, p^*)$. In Tab.\ref{tab:comparison_diffpoints}(b), we took $\mu^* = 2, \, p^* = 0.2$, which is the lowest corner of $\mathcal E$, while in Tab.\ref{tab:comparison_diffpoints}(c), $\mu^* = 1.2, \, p^* = 0.5$ is inside $\mathcal E$ but not far away from the boundary.These numbers show that $\beta_{fw}(M_1, m)$ decreases rapidly when $\mu^*$ and/or $p^*$ increases. 
 
\begin{table}[ht]
\begin{center}
\begin{subtable}{.43\linewidth}
\centering
\begin{tabular}{c c| c c c c}
 & \multicolumn{1}{c}{}& \multicolumn{4}{c}{$M_1$} \\ \cline{3-6}
 & & 1 & 2 & 3 & 4\\ \cline{2-6}
\multicolumn{1}{c|}{\multirow{4}{*}{$m$}} & 1 & 0.305 &  0.469  & 0.534  & 0.200 \\
\multicolumn{1}{c|}{} & 2 &  & 0.305 &  0.469  & 0.534 \\ 
\multicolumn{1}{c|}{} & 3 &  &  & 0.305 &  0.469 \\
\multicolumn{1}{c|}{} & 4 &  &  &  & 0.305  \\ \cline{3-6}
\end{tabular}\vspace{0.1cm}
\caption{Upper bound on $\beta^{se}_{\text{fw}}(M_1,m)$.}
\end{subtable}\\
\begin{subtable}{.43\linewidth}
\centering
\begin{tabular}{c c| c c c c}
 & \multicolumn{1}{c}{}& \multicolumn{4}{c}{$M_1$} \\ \cline{3-6}
 & & 1 & 2 & 3 & 4\\ \cline{2-6}
\multicolumn{1}{c|}{\multirow{4}{*}{$m$}} & 1 & 0.303 & 0.464 & 0.515 & 0.200 \\
\multicolumn{1}{c|}{} & 2 &  & 0.091 & 0.170 & 0.099 \\ 
\multicolumn{1}{c|}{} & 3 &  &  & 0.026 & 0.030\\
\multicolumn{1}{c|}{} & 4 &  &  &  & 0.004  \\ \cline{3-6}
\end{tabular}\vspace{0.1cm}
\caption{$\beta_{\text{fw}}(M_1,m)$ for $(\mu^*, p^*) = (2, 0.2)$.}
\end{subtable}
\hspace{0.8cm}
\begin{subtable}{.43\linewidth}
\centering
\begin{tabular}{c c| c c c c}
 & \multicolumn{1}{c}{}& \multicolumn{4}{c}{$M_1$} \\ \cline{3-6}
 & & 1 & 2 & 3 & 4\\ \cline{2-6}
\multicolumn{1}{c|}{\multirow{4}{*}{$m$}} & 1 &  0.032 & 0.050 & 0.050 & 0.002\\ 
\multicolumn{1}{c|}{} & 2 &  & 0.001 & 0.002 & $<0.001$\\
\multicolumn{1}{c|}{} & 3 &  &  & $<0.001$ & $<0.001$\\
\multicolumn{1}{c|}{} & 4 &  &  &  & $<0.001$ \\ \cline{3-6}
\end{tabular}\vspace{0.1cm}
\caption{$\beta_{\text{fw}}(M_1,m)$ for $(\mu^*, p^*) = (1.2, 0.5)$.}
\end{subtable}
\end{center}
\caption{$\beta_{\text{fw}}(M_1, m)$. The design parameters are $\alpha = 0.05, \, \beta_{\max} = 0.2,\,n_1 = 100,\,\alpha_0 = 0.7, \,\alpha_1 = 0.026,\,n_2 = 65$. Real parameters in the centers with strong effect are $(\mu^*,\, p^*)$, for the other centers $\mathcal H_0$ is true.}
\label{tab:comparison_diffpoints}
\end{table}

\section{Examples}\label{Simulations}


To validate the model, we take the parameters from the previous section: $M = 4, \, \alpha = 0.05, \,\beta_{\max} = 0.2,\,\, n_1 = 100,\,\alpha_0 = 0.7, \,\alpha_1 = 0.026,\,n_2 = 65$. For different $M_1$ and $m$ we simulate the trial 1000 times. The responses in the control and treatment groups are, $Z^C \sim \mathcal N(0, 1)$, $Z^T\sim (1-p^*)\mathcal N(0,1) + p^*\mathcal N(\mu^*, 1)$. We estimate the standard deviation as $\hat\sigma = \sqrt{\frac{\sum_{i=1}^n (Z^C_i - \hat\mu^C)^2}{n-1}}$. The results are summarized in Tab.\ref{tab:sim_Mgroups}.
\begin{table}[ht]
\begin{subtable}{.42\linewidth}
\centering
\begin{tabular}{cc|c|c|c|c|} 
 & \multicolumn{1}{c}{}& \multicolumn{4}{c}{$M_1$} \\ \cline{3-6}
  & & 1 & 2 & 3 & 4 \\ \cline{2-6}
\multicolumn{1}{c|}{\multirow{8}{*}{$m$}} & \multicolumn{1}{c|}{\multirow{2}{*}{1}} &  0.303 & 0.464 & 0.515 & 0.200\\
\multicolumn{1}{c|}{} & \multicolumn{1}{c|}{} & 0.294 & 0.498 & 0.523 & 0.185 \\ \cline{2-6}
\multicolumn{1}{c|}{} & \multicolumn{1}{c|}{\multirow{2}{*}{2}} &  & 0.091 & 0.170 & 0.100 \\
\multicolumn{1}{c|}{} & \multicolumn{1}{c|}{} & & 0.099 & 0.179 & 0.079 \\ \cline{2-6}
\multicolumn{1}{c|}{} & \multicolumn{1}{c|}{\multirow{2}{*}{3}} &  &  & 0.026 & 0.030\\
\multicolumn{1}{c|}{} & \multicolumn{1}{c|}{} & & & 0.032 & 0.032\\ \cline{2-6}
\multicolumn{1}{c|}{} & \multicolumn{1}{c|}{\multirow{2}{*}{4}} &  &  &  & 0.004 \\
\multicolumn{1}{c|}{} & \multicolumn{1}{c|}{} & & & & 0.007\\ \cline{2-6}
\end{tabular}
\caption{For the centres with strong effect $(\mu^*, p^*) = (2, 0.2)$.}
\end{subtable}\hspace{0.8cm}
\begin{subtable}{.42\linewidth}
\centering
\begin{tabular}{cc|c|c|c|c|} 
 & \multicolumn{1}{c}{}& \multicolumn{4}{c}{$M_1$} \\ \cline{3-6}
  & & 1 & 2 & 3 & 4 \\ \cline{2-6}
\multicolumn{1}{c|}{\multirow{8}{*}{$m$}} & \multicolumn{1}{c|}{\multirow{2}{*}{1}} & 0.032 & 0.050 & 0.050 & 0.002\\
\multicolumn{1}{c|}{} & \multicolumn{1}{c|}{} & 0.041 & 0.044 & 0.043 & 0.005\\ \cline{2-6}
\multicolumn{1}{c|}{} & \multicolumn{1}{c|}{\multirow{2}{*}{2}} &  & 0.001 & 0.002 & $<0.001$ \\
\multicolumn{1}{c|}{} & \multicolumn{1}{c|}{} &  & 0.001 & 0 & 0 \\ \cline{2-6}
\multicolumn{1}{c|}{} & \multicolumn{1}{c|}{\multirow{2}{*}{3}} &  &  & $<0.001$ & $<0.001$\\
\multicolumn{1}{c|}{} & \multicolumn{1}{c|}{} &  & & 0 & 0\\ \cline{2-6}
\multicolumn{1}{c|}{} & \multicolumn{1}{c|}{\multirow{2}{*}{4}} &  &  &  & $<0.001$ \\
\multicolumn{1}{c|}{} & \multicolumn{1}{c|}{} & & &  & 0\\ \cline{2-6}
\end{tabular}
\caption{For the centres with strong effect $(\mu^*, p^*) = (1.2, 0.5)$.}
\end{subtable}
 \caption{Theoretical (top) and empirical (bottom) $\beta_{\text{fw}}(M_1, m)$ for 1000 simulations. The parameters are $M=4, \, \alpha = 0.05, \, \beta_{\max} = 0.2,\,n_1 = 100,\,\alpha_0 = 0.7, \,\alpha_1 = 0.026,\,n_2 = 65$.}
\label{tab:sim_Mgroups}
\end{table}

\section{Discussion}\label{disc}
\subsection*{Region of strong effect}
The region of strong effect for the mixture response model was introduced as an instrument for decision making. Further standardization for regulatory reasons might be of interest. From our point of view, it may also assist in making the research on trial designs more coherent. In our framework we suggest to first determine the region of strong effect. while $\beta_{\max}$ and other parameters may be tuned according to the needs of the study.
\subsection*{Mean value statistic}
The choice of the mean value statistic is not only motivated by its simplicity and its role in the classical RCT design. The mean value statistic also has some advantages for the mixture models. If responses in the control and treatment groups, $Y^C$ and $Y^T$, have densities $f(y)$ and $(1-p)f(y) + pf(y-\mu)$, respectively, and $\Var\left(Y^C\right) = 1$, then by the central limit theorem the following approximations hold
$$\overline{Y^T} \sim \mathcal N\left(\mu p, \frac{1+(1-p)p\mu^2}{n}\right),\hspace{0.2cm} \overline X \sim \mathcal N\left(\mu p, \frac{1+(1-p)p\mu^2}{n}+\frac{1}{n}\right).$$
This means that all the designs presented in the paper may be used for non-normal responses.
\par In the two-stage design we use the criterion $\overline X>\eta_2$ to decide on rejection of the null after the second stage. It is not obvious why this test works well and why $\overline X_1$ and $\overline X_2$ are not combined in some other way, but the proposed test is close to the UMP. The likelihood ratio based on the joint distribution function of $\left(\overline X_1, \overline X_2\right)$ is
\begin{equation}
\begin{aligned}
& \frac{g_{\left(\overline{X_1},\overline{X_2}\right)}(x_1,x_2)}{f_{\left(\overline{X_1},\overline{X_2}\right)}(x_1,x_2)}=\frac{\varphi\left(\sqrt{\frac{n_1}{2}}\frac{x_1 - \mu p}{\sqrt{1+\frac{(1-p)p}{2}\mu^2}}\right)\varphi\left(\sqrt{\frac{n_2}{2}} \frac{x_2 - \mu p}{\sqrt{1+\frac{(1-p)p}{2}\mu^2}}\right)}{\varphi\left(\sqrt{\frac{n_1}{2}}\,x_1\right)\varphi\left(\sqrt{\frac{n_2}{2}}\,x_2\right)}\\
& \approx \frac{\varphi\left(\sqrt{\frac{n_1}{2}}(x_1 - \mu p)\right)\varphi\left(\sqrt{\frac{n_2}{2}}(x_2 - \mu p)\right)}{\varphi\left(\sqrt{\frac{n_1}{2}}\,x_1\right)\varphi\left(\sqrt{\frac{n_2}{2}}\,x_2\right)} = e^{\frac{(x_1 n_1+x_2 n_2)\mu p-\frac{1}{2}(\mu p)^2}{2}},
\end{aligned}
\end{equation}
where $f(x_1, x_2)$ and $g(x_1, x_2)$ are pdfs under the null and alternative, respectively. The normal approximation in the likelihood works if the variance $1+\frac{(1-p)p\mu^2}{2}$ is close to one. In this case likelihood ratio is the monotone non-decreasing function of $\frac{x_1 n_1+x_2 n_2}{n_1+n_2}$. The test becomes UMP when $p=1$.

\subsection*{Expected sample size}
As we have seen in Sec.\ref{Sec:Planning}, the minimum value of the optimized $q_{1}$ is quite close to $n$. This value is reached for the 'worst' alternative $\mu = \frac{z_{1-\alpha_1}+z_{\alpha_0}}{\sqrt {2 n_1}}, \,p=1$ and for other values it can be substantially lower. And in many situations we have much smaller expected sample size. To illustrate this, we plot the probability to conduct the second stage for the chosen in Sec.\ref{Sec:Planning} single- and multicenter designs (see Fig.\ref{fig:prob_2stage}).
\begin{figure}
\begin{center}
\begin{subfigure}{.45\textwidth}
\includegraphics[width =  \textwidth]{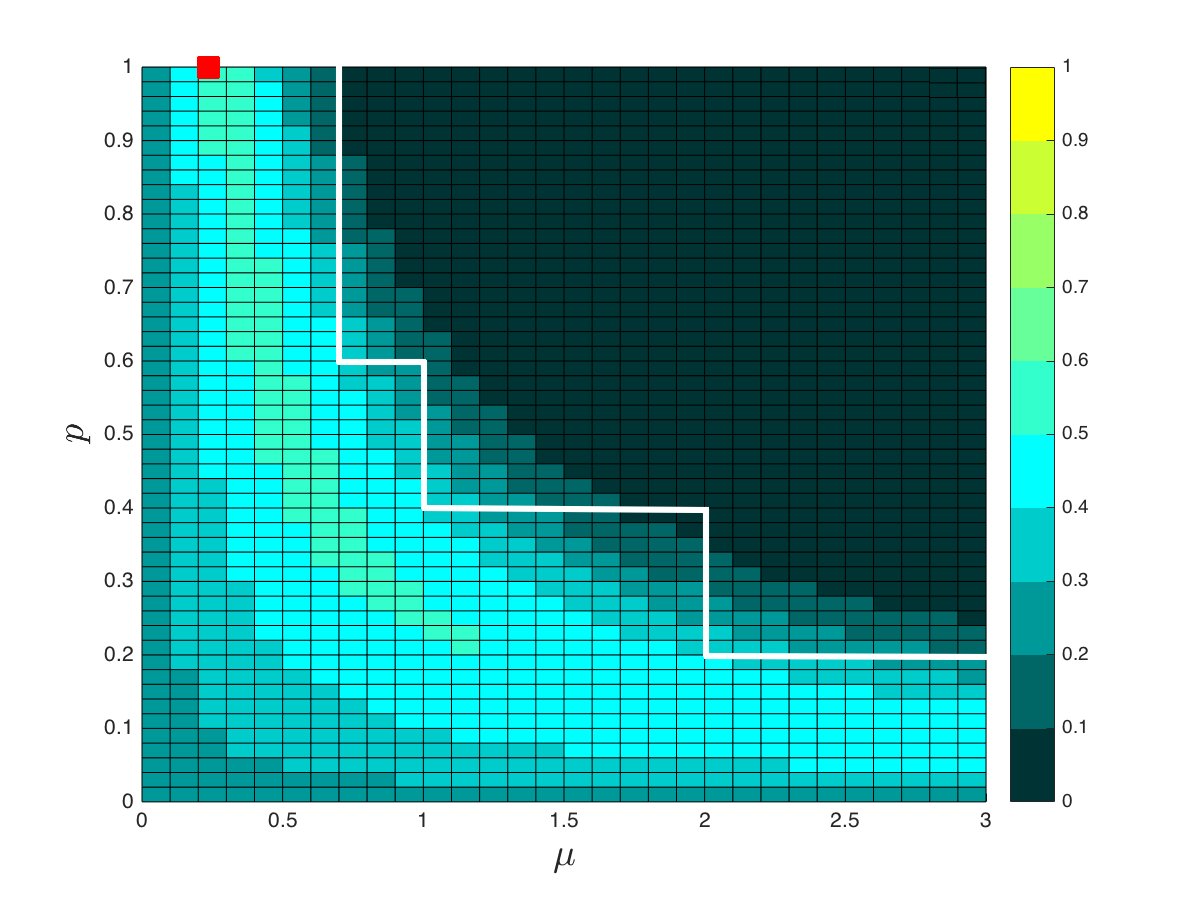}
\caption{Single center trial. The maximum probability is $0.52$.}
\end{subfigure}
\begin{subfigure}{.45\textwidth}
\includegraphics[width = \textwidth]{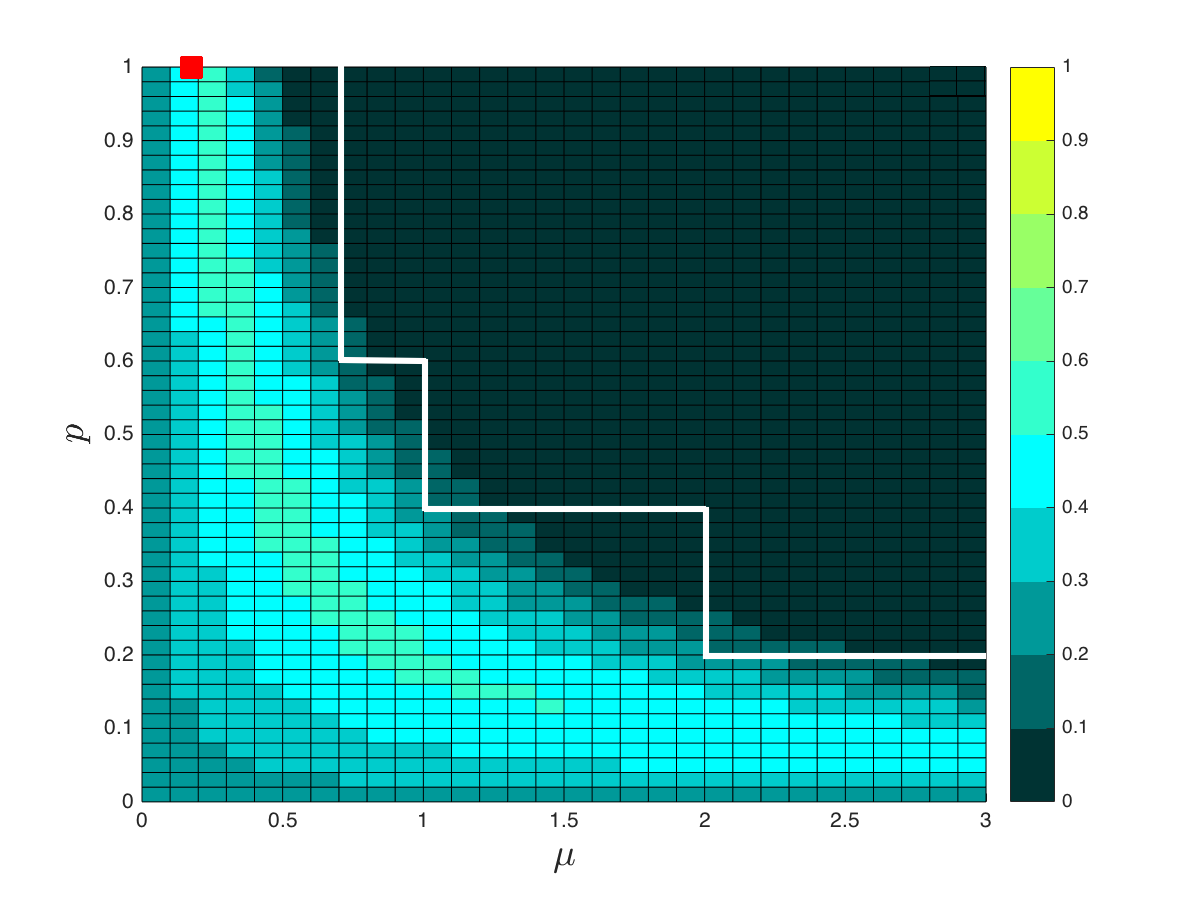}
\caption{Multicenter trial. The maximum probability is $0.52$.}
\end{subfigure}
\caption{\hspace{0.2cm}The plot is in the $(\mu,\,p)$ space and shows the region of strong effect together with the probability of conducting the second stage. On the left is the case of a single center RCT and on the right is the case of 4 centers RCT. In each case the red dot, $(0.24, 1)$ on the left  and $(0.17, 1)$ on the right, shows the worst combination of $\mu$ and $p$ in terms of the expected sample size under the alternative. The design on the left is $n_1 = 55, \, \alpha_0 = 0.7, \alpha_1 = 0.026, \, n_2 = 38$, whereas on the right it has $M = 4, \, n_1 = 100, \, \alpha_0 = 0.7, \alpha_1 = 0.026, \, n_2 = 65$.}
\label{fig:prob_2stage}
\end{center}
\end{figure}
This graph might be also of interest for the user during the design planning.

\subsection*{Rejection procedures for multicenter trials}
For multicenter RCTs, we construct the design using Hochberg's step-up rejection procedure. However, one can use any step-up procedure with thresholds $\alpha(1), ..., \alpha(M)$. For example, it could be Benjamini-Hochberg's procedure (~\cite{Benjamini:1995aa}) with $\alpha(k)=\frac{k\alpha}{M}$ which controls the false discovery rate (FDR). There will be no difference in parameters $\{n_1, \alpha_0, \alpha_1, n_2\}$ for the designs using Hochberg's or Benjamini-Hochberg's procedures to control $\beta^{se}_{\text{fw}}(M,1)$, because in each center the design corresponds to the single center RCT with the maximum type I error rate controlled at the same level $\alpha(M) = \alpha$ (see \eqref{ind:errors} and Lemma \ref{lemma:multi_beta_max}). If $M=2$, these methods are fully equivalent, including the thresholds $\alpha(1) = \alpha/2, \, \alpha(2) = \alpha$. Nonetheless, if $M>2$, $\beta^{se}_{\text{fw}}(M_1, m)$ is generally smaller for the Benjamini-Hochberg procedure (see Tab.\ref{tab:comparison_BH_H}).


\begin{table}[ht]
\begin{subtable}{.42\linewidth}
\centering
\begin{tabular}{cc|c|c|c|c|} 
 & \multicolumn{1}{c}{}& \multicolumn{4}{c}{$M_1$} \\ \cline{3-6}
  & & 1 & 2 & 3 & 4 \\ \cline{2-6}
\multicolumn{1}{c|}{\multirow{8}{*}{$m$}} & \multicolumn{1}{c|}{\multirow{2}{*}{1}} &  0.303 & 0.464 & 0.515 & 0.200\\
\multicolumn{1}{c|}{} & \multicolumn{1}{c|}{} & 0.298 & 0.380 & 0.200 & 0.200 \\ \cline{2-6}
\multicolumn{1}{c|}{} & \multicolumn{1}{c|}{\multirow{2}{*}{2}} &  & 0.091 & 0.170 & 0.099 \\
\multicolumn{1}{c|}{} & \multicolumn{1}{c|}{} & & 0.082 & 0.069 & 0.027 \\ \cline{2-6}
\multicolumn{1}{c|}{} & \multicolumn{1}{c|}{\multirow{2}{*}{3}} &  &  & 0.026 & 0.030\\
\multicolumn{1}{c|}{} & \multicolumn{1}{c|}{} & & & 0.014 & 0.009\\ \cline{2-6}
\multicolumn{1}{c|}{} & \multicolumn{1}{c|}{\multirow{2}{*}{4}} &  &  &  & 0.004 \\
\multicolumn{1}{c|}{} & \multicolumn{1}{c|}{} & & & & 0.002\\ \cline{2-6}
\end{tabular}
\caption{For the centres with strong effect $(\mu^*, p^*) = (2, 0.2)$.}
\end{subtable}\hspace{0.8cm}
\begin{subtable}{.42\linewidth}
\centering
\begin{tabular}{cc|c|c|c|c|} 
 & \multicolumn{1}{c}{}& \multicolumn{4}{c}{$M_1$} \\ \cline{3-6}
  & & 1 & 2 & 3 & 4 \\ \cline{2-6}
\multicolumn{1}{c|}{\multirow{8}{*}{$m$}} & \multicolumn{1}{c|}{\multirow{2}{*}{1}} & 0.032 & 0.050 & 0.050 & 0.002\\
\multicolumn{1}{c|}{} & \multicolumn{1}{c|}{} & 0.032 & 0.033 & 0.003 & 0.002\\ \cline{2-6}
\multicolumn{1}{c|}{} & \multicolumn{1}{c|}{\multirow{2}{*}{2}} &  & 0.001 & 0.002 & $<0.001$ \\
\multicolumn{1}{c|}{} & \multicolumn{1}{c|}{} &  & $<0.001$ & $<0.001$ & $<0.001$ \\ \cline{2-6}
\multicolumn{1}{c|}{} & \multicolumn{1}{c|}{\multirow{2}{*}{3}} &  &  & $<0.001$ & $<0.001$ \\
\multicolumn{1}{c|}{} & \multicolumn{1}{c|}{} &  & & $<0.001$ & $<0.001$ \\ \cline{2-6}
\multicolumn{1}{c|}{} & \multicolumn{1}{c|}{\multirow{2}{*}{4}} &  &  &  & $<0.001$ \\
\multicolumn{1}{c|}{} & \multicolumn{1}{c|}{} & & &  & $<0.001$ \\ \cline{2-6}
\end{tabular}
\caption{For the centres with strong effect $(\mu^*, p^*) = (1.2, 0.5)$.}
\end{subtable}
 \caption{$\beta_{\text{fw}}(M_1, m)$ for Hochberg's (top) and Benjamini-Hochberg's (bottom) rejection rules. The parameters are $M = 4, \, \alpha = 0.05, \, \beta_{\max} = 0.2,\,n_1 = 100,\,\alpha_0 = 0.7, \,\alpha_1 = 0.026,\,n_2 = 65$.}
\label{tab:comparison_BH_H}
\end{table}
To emphasize the importance of step-up procedures, consider the classical Bonferroni correction. For this procedure the decision is made for each center individually. However, the method may be reformulated as a step-up procedure with $\alpha(k)=\frac{\alpha}{M}$ for $k = 1, ..., M$. An additional constraint for the two-stage design will be $\alpha_1<\frac{\alpha}{M}$. If this is not satisfied, the second stage is not needed. In Fig.\ref{tab:multipl_comp_designs}, we give a diagnostic plot similar to those in Fig.\ref{n1alpha0_mult} for the Bonferroni method. First, the sample size for the one-stage design is $n = 209$, while for Hochberg's and Benjamini-Hochberg's procedures $n = 153<209$.

\begin{figure} 
\begin{center}
\begin{subfigure}{.6\textwidth}
\includegraphics[width = \textwidth]{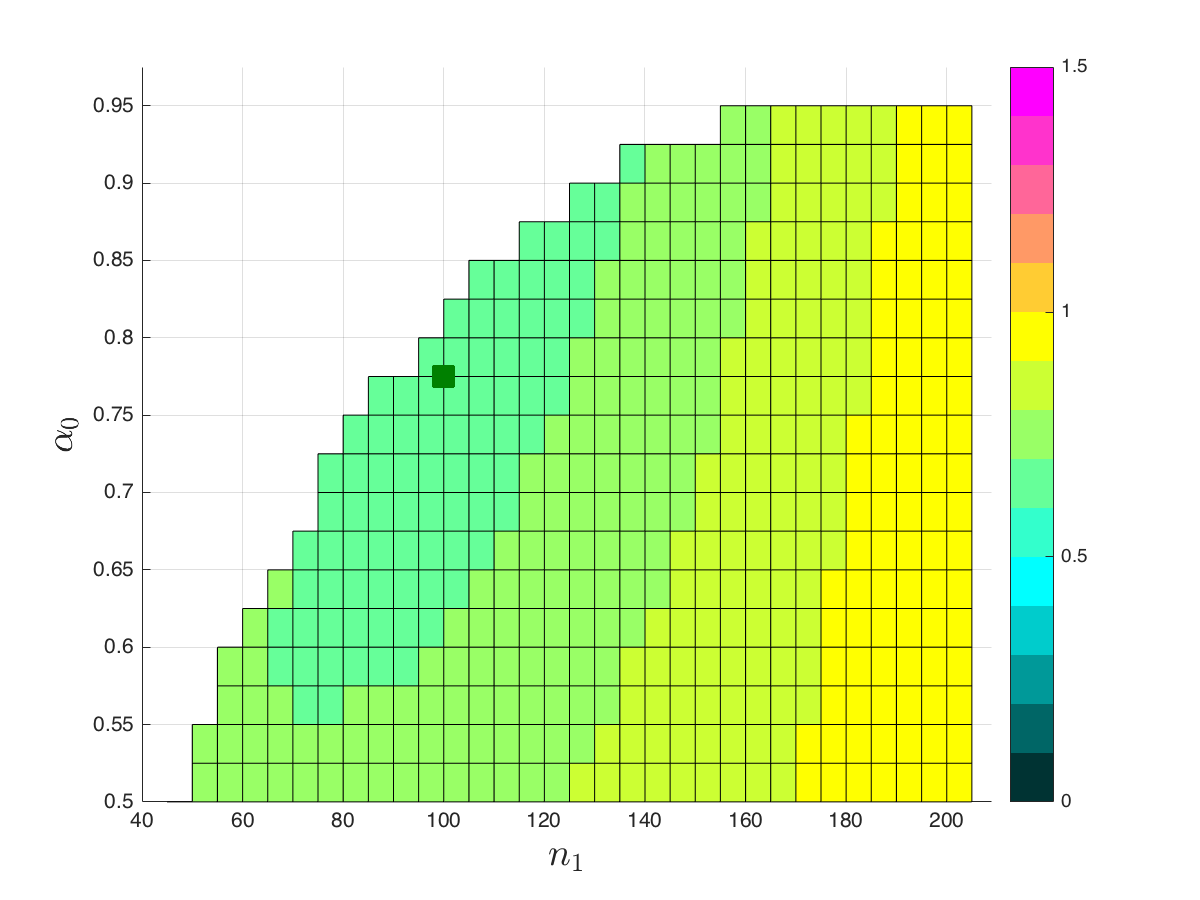}
\caption{$\frac{q_{0}}{n}$. \\ \vspace{0.2cm} $\min q_{0}  = 134$ for $n_1 = 100, \, \alpha_0 = 0.775$.}
\end{subfigure}
\begin{subfigure}{.6\textwidth}
\includegraphics[width = \textwidth]{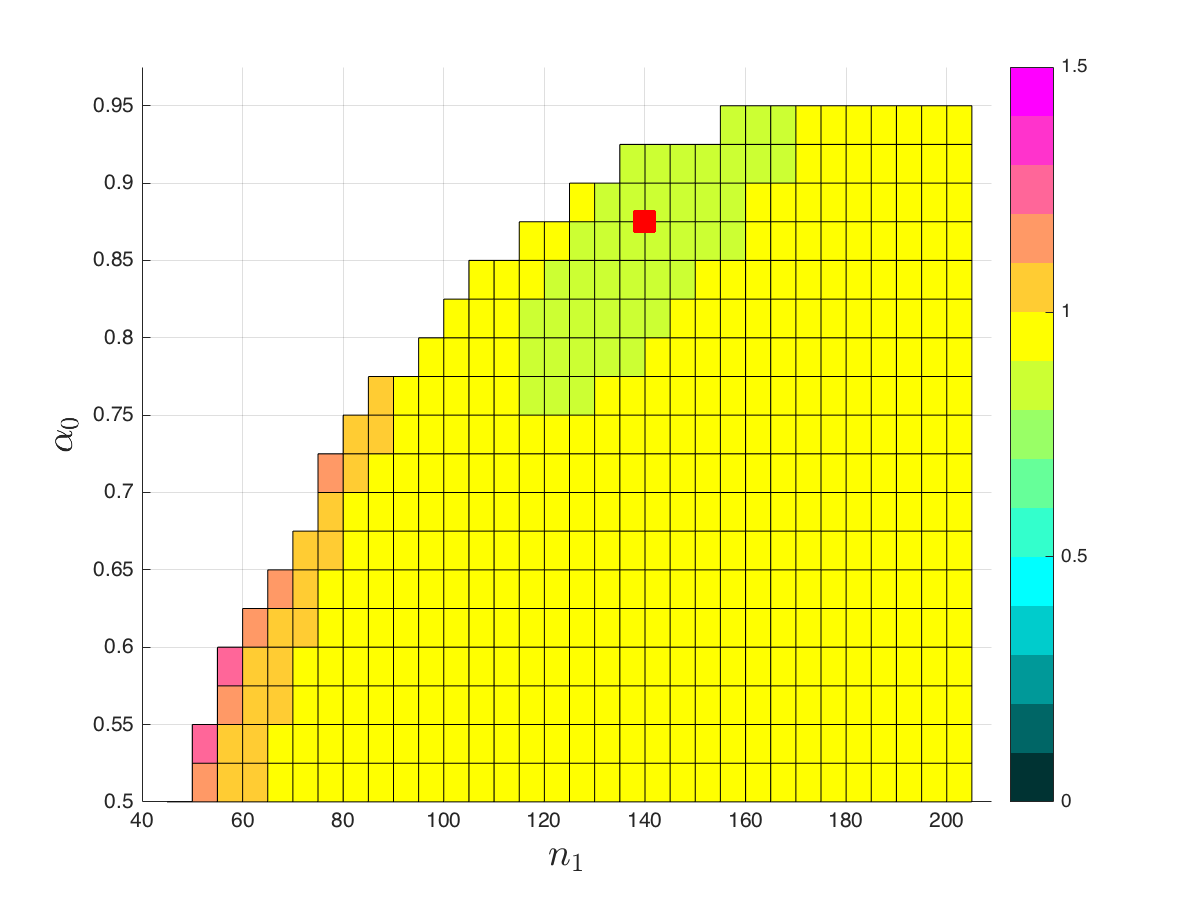}
\caption{$\frac{q_{1}}{n}$. \\ \vspace{0.2cm}$ \min q_{1} = 185$ for $n_1 = 140, \, \alpha_0 = 0.875$.}
\end{subfigure}\\
\begin{subfigure}{.6\textwidth}
\includegraphics[width = \textwidth]{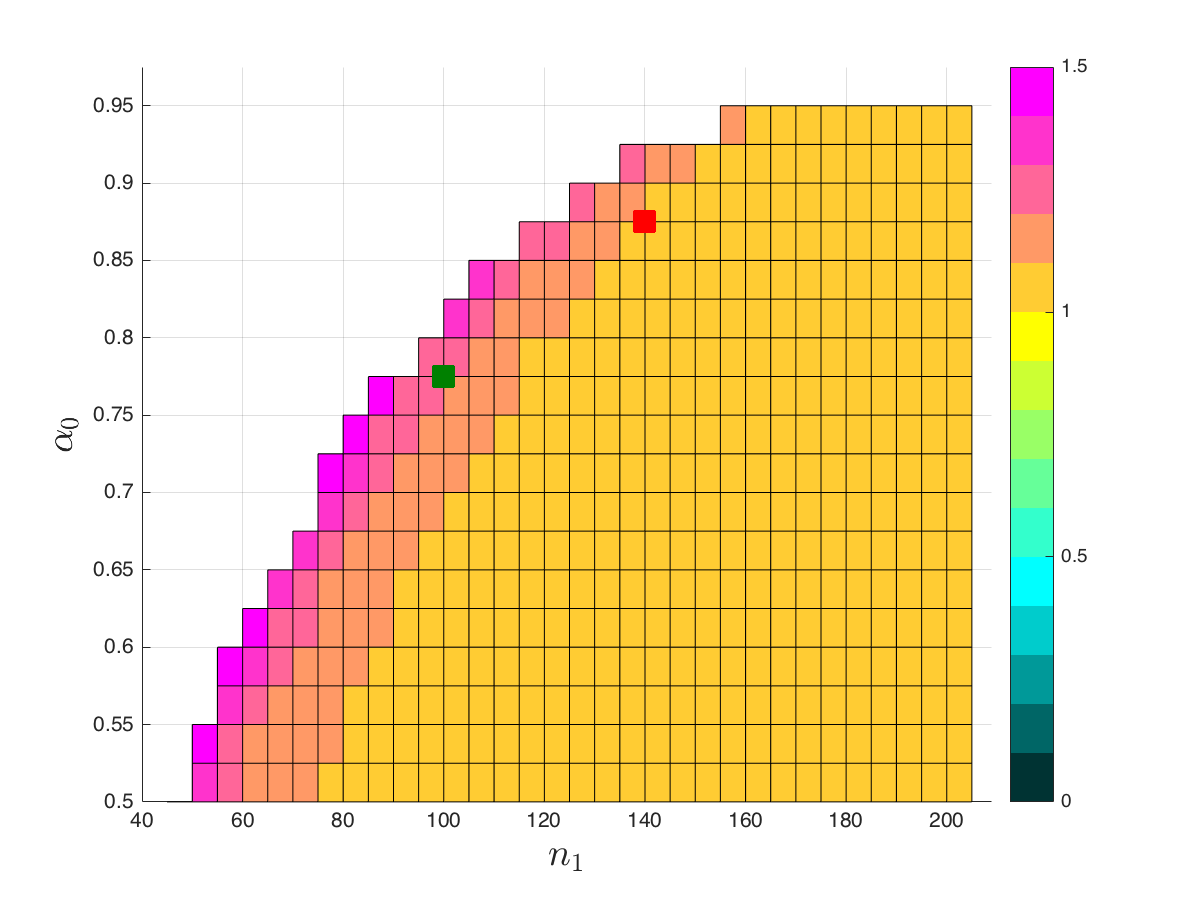}
\caption{$\frac{n_1+n_2}{n}$.\\ \vspace{0.2cm} $\min (n_1+n_2) = 209$ for $n_1 = 205, \, \alpha_0 = 0.55$.}
\end{subfigure}
\caption{\hspace{0.2cm}The diagnostic plot for choosing the design. The parameters are $M=4, \, \alpha = 0.05, \, \beta_{\max} = 0.2, \, \vec \mu = (2, 1, 0.7), \, \vec p = (0.2,0.4,0.6)$. For the Bonferroni procedure $n = 209$.}
\label{tab:multipl_comp_designs}
\end{center}
\end{figure}
The minimum values of $q_{0},\, q_{1},\,n_1+n_2$ for the Bonferroni method are also considerably higher than for the step-up procedures ($q_{0} = 118,\, q_{1} = 134, \, n_1+n_2 = 165$). Therefore,  the designs using Hochberg's step-up procedure controls the same FWER as the Bonferroni and is clearly superior in terms of effort.

\subsection*{Random treatment effect}
In our model the treatment-specific effect for the drug responders is a fixed value $\mu^T-\mu^C$. However, a more realistic model might consider it to be random. To address this issue, we present type II error rates for the chosen designs if the treatment-specific effect is random variable from $\mathcal N(\mu^T-\mu^C, \delta^2)$. There, $Z^T\sim(1-p)\mathcal N(\mu^C, \sigma^2) + p\mathcal N(\mu^T, \sigma^2+\delta^2)$. In Tab.\ref{tab:random_effect}, average empirical type II error rates for the fixed and random treatment-specific effects are given. We observe bigger error rates when an effect is random. This might be taken into account for the development of the design.
\begin{table}[ht]
\begin{subtable}{.42\linewidth}
\centering
\begin{tabular}{cc|c|c|c|c|} 
 & \multicolumn{1}{c}{}& \multicolumn{4}{c}{$M_1$} \\ \cline{3-6}
  & & 1 & 2 & 3 & 4 \\ \cline{2-6}
\multicolumn{1}{c|}{\multirow{8}{*}{$m$}} & \multicolumn{1}{c|}{\multirow{2}{*}{1}} &  0.294 & 0.498 & 0.523 & 0.185\\
\multicolumn{1}{c|}{} & \multicolumn{1}{c|}{} & 0.321 & 0.519 & 0.571 & 0.355 \\ \cline{2-6}
\multicolumn{1}{c|}{} & \multicolumn{1}{c|}{\multirow{2}{*}{2}} &  & 0.099 & 0.179 & 0.079 \\
\multicolumn{1}{c|}{} & \multicolumn{1}{c|}{} & & 0.111 & 0.22 & 0.193 \\ \cline{2-6}
\multicolumn{1}{c|}{} & \multicolumn{1}{c|}{\multirow{2}{*}{3}} &  &  & 0.032 & 0.032\\
\multicolumn{1}{c|}{} & \multicolumn{1}{c|}{} & & & 0.038 & 0.068\\ \cline{2-6}
\multicolumn{1}{c|}{} & \multicolumn{1}{c|}{\multirow{2}{*}{4}} &  &  &  & 0.007 \\
\multicolumn{1}{c|}{} & \multicolumn{1}{c|}{} & & & & 0.009\\ \cline{2-6}
\end{tabular}
\caption{For the centres with strong effect $(\mu^*, p^*) = (2, 0.2)$.}
\end{subtable}\hspace{0.8cm}
\begin{subtable}{.42\linewidth}
\centering
\begin{tabular}{cc|c|c|c|c|} 
 & \multicolumn{1}{c}{}& \multicolumn{4}{c}{$M_1$} \\ \cline{3-6}
  & & 1 & 2 & 3 & 4 \\ \cline{2-6}
\multicolumn{1}{c|}{\multirow{8}{*}{$m$}} & \multicolumn{1}{c|}{\multirow{2}{*}{1}} & 0.041 & 0.044 & 0.043 & 0.005\\
\multicolumn{1}{c|}{} & \multicolumn{1}{c|}{} & 0.150 & 0.279 & 0.353 & 0.253 \\ \cline{2-6}
\multicolumn{1}{c|}{} & \multicolumn{1}{c|}{\multirow{2}{*}{2}} &  & 0 & 0 & 0 \\
\multicolumn{1}{c|}{} & \multicolumn{1}{c|}{} &  & 0.038 & 0.075 & 0.072 \\ \cline{2-6}
\multicolumn{1}{c|}{} & \multicolumn{1}{c|}{\multirow{2}{*}{3}} &  &  & 0 & 0\\
\multicolumn{1}{c|}{} & \multicolumn{1}{c|}{} &  & & 0.008 & 0.012\\ \cline{2-6}
\multicolumn{1}{c|}{} & \multicolumn{1}{c|}{\multirow{2}{*}{4}} &  &  &  & 0 \\
\multicolumn{1}{c|}{} & \multicolumn{1}{c|}{} & & &  & 0 \\ \cline{2-6}
\end{tabular}
\caption{For the centres with strong effect $(\mu^*, p^*) = (1.2, 0.5)$.}
\end{subtable}
 \caption{Empirical type II error rates based on 1000 simulations. $Z^C\sim\mathcal N(0, 1)$, $Z^T\sim (1-p^*)\mathcal N(0,1) + p^*\mathcal N(\mu^*, 1)$ for fixed effect (top) and $Z^T\sim (1-p^*)\mathcal N(0,1) + p^*\mathcal N\left(\mu^*, 1+0.5^2\right)$ for random effect (bottom) effects. The parameters are $M=4, \, \alpha = 0.05, \, \beta_{\max} = 0.2,\,n_1 = 100,\,\alpha_0 = 0.7, \,\alpha_1 = 0.026,\,n_2 = 65$.}
\label{tab:random_effect}
\end{table}

\subsection*{Estimation of $\mu^T, \, \mu^C, \, \sigma$ and $p$}
For the classical sequential RCT design (no mixture response) the estimation of the parameters is discussed in \cite{Armitage:1969aa, Siegmund:1978aa, OBrien:1979aa}, and \cite{Tsiatis:1984aa}. In the mixture framework, to characterize the sensitive subgroup, one would like to have estimators of $\mu^T, \, \mu^C, \, \sigma$ and $p$. We do not further investigate this question, but there are some estimation procedures that can be applied to the data. A large literature covers various aspects of estimation of the mixture model parameters. In the case of a one-stage design, one can use the likelihood-based estimators of $\mu^T, \, \mu^C, \, \sigma$ and $p$ (e.g., the EM algorithm or the method proposed in \cite{Pavlic:2001aa}). Regarding a two-stage trial, one can apply these algorithms to either the first or the second stage results.

\section{Conclusion}

\par In this paper we extend the classical continuous response RCT design to the case when a fraction of the treated patients responds to the tested treatment. The response in the treatment group is modelled as a two-component mixture, representing placebo and drug responders. 
\par We modified the conventional procedure that decides on the existence of a treatment-specific effect to a test that decides on the existence of a drug-responders subgroup. We characterize the latter by the fraction $p>0$ and the treatment-specific effect $\mu>0$. We assume that pairs $(\mu, p)$ of potential interest belong to the set $\cal E$, the region of strong effect. The trial protocol is then determined by the testing procedure that ensures a certain power of detection for any subgroup whose parameter is in $\cal E$. We developed one- and  two-stage RCT designs along with the exact and approximate expressions for the type II error rates. We also generalized these designs to the multicenter framework where we control the family-wise error rate. To decrease the sample size, we use Hochberg's step-up multiple test. If one wants to control FDR, one may simply use Benjamini--Hochberg's test or any other suitable thresholds. We provide the graphical help to guide the user in choosing design parameters, where a few are determined by minimizing the required effort.

\appendix

\section{Mathematical proofs}\label{App:Mathematical proofs}
\setcounter{theorem}{0}

\begin{lemma}
The tests that reject $\mathcal H_0$ for $\overline X>\eta$ are uniformly most powerful.
\end{lemma}
\begin{proof}
We denote the pdf of $\overline X$ under $\mathcal H_0$ as $f(x)$ and the pdf of $\overline X$ under an alternative as $g(x)$.
The function $\frac{g(x)}{f(x)}$ is non-decreasing in $x$ if $\mu>0$, since
\begin{equation*}
\begin{aligned}
& \frac{g(x)}{f(x)} = \sum_{k=0}^n \binom {n} {k}p^k(1-p)^{n-k}\frac{f\left(x-\frac{k}{n}\mu\right)}{f(x)} = \sum_{k=0}^n \binom {n} {k}p^k(1-p)^{n-k} e^{\frac{n}{2}\left(x \frac{k}{n}\mu -\frac{1}{2}\left(\frac{k}{n}\mu\right)^2\right)}.
\end{aligned}
\end{equation*}
From the Karlin--Rubin theorem it follows that the test $\overline X > \eta$ is UMP.
\end{proof}


\begin{lemma}
For a region of strong effect as defined in Fig.\ref{fig:region_of_interest}, the maximum type II error rate in the region of the strong effect satisfies 
\begin{equation}
\beta^{se}(n,\alpha) =\underset{i=1, ..., s}{\max}\beta(n,z_{1-\alpha}\sqrt{\frac{2}{n}},\mu_i,p_i)
\end{equation}
\end{lemma}
\begin{proof}
The probability of a false negative is
$$\beta(n,\eta,\mu,p)=\sum_{k=0}^n \binom {n} {k}p^k(1-p)^{n-k}\Phi\left(\sqrt{\frac{n}{2}}\left(\eta-\frac{k}{n}\mu\right)\right).$$
The maximum value is attained in one of the corners of $\cal E$. Indeed, this follows from $$\frac{\partial \beta}{\partial \mu}(n,\eta,\mu,p)<0 \hspace{0.2cm} \text{and} \hspace{0.2cm}\frac{\partial \beta}{\partial p}(n,\eta,\mu,p)<0.$$
\begin{equation}\label{eq:beta_mu}
\begin{aligned}
&\frac{\partial \beta}{\partial \mu}(n,\eta,\mu,p) = -\sum_{k=0}^n \binom {n} {k}p^k(1-p)^{n-k}\frac{k}{\sqrt{2n}}\varphi\left(\sqrt{\frac{n}{2}}\left(\eta-\frac{k}{n}\mu\right)\right) < 0;\\
\end{aligned}
\end{equation}
\begin{equation}\label{eq:beta_p}
\begin{aligned}
&\frac{\partial \beta}{\partial p}(n,\eta,\mu,p) = n\sum_{k=1}^n \binom {n-1} {k-1}p^{k-1}(1-p)^{n-k}\Phi\left(\sqrt{\frac{n}{2}}\left(\eta-\frac{k}{n}\mu\right)\right) \\
&- n\sum_{k=0}^{n-1} \binom {n-1} {k}p^{k}(1-p)^{n-k-1}\Phi\left(\sqrt{\frac{n}{2}}\left(\eta-\frac{k}{n}\mu\right)\right) \\
&=n\sum_{k=0}^{n-1} \binom {n-1} {k}p^{k}(1-p)^{n-k-1}\left(\Phi\left(\sqrt{\frac{n}{2}}\left(\eta-\frac{k+1}{n}\mu\right)\right) - \Phi\left(\sqrt{\frac{n}{2}}\left(\eta-\frac{k}{n}\mu\right)\right)\right) < 0.
\end{aligned}
\end{equation}
\end{proof}


\begin{lemma}
For the region of strong effect defined as in Definition \ref{def:se}, the maximum type II error rate for the two-stage trial is
$$\beta_2^{se}(D,\alpha)=\underset{i=1, \ldots, s}{\max}\beta_2(D,\alpha,\mu_i,p_i).$$
\end{lemma}
\begin{proof}
This can be shown by proving that
$$\frac{\partial \beta_2}{\partial \mu}(n_1,\eta_0,\eta_1,n_2,\alpha,\mu,p)<0 \text{ and } \frac{\partial \beta_2}{\partial p}(n_1,\eta_0,\eta_1,n_2,\alpha,\mu,p)<0.$$
We will use inequalities $\frac{\partial \beta}{\partial \mu}(n, \eta, \mu, p)<0$ and $\frac{\partial \beta}{\partial p}(n, \eta, \mu, p)<0$ from Lemma~\ref{lemma:beta_max} and inequality $\frac{\partial \beta}{\partial \eta}(n, \eta, \mu, p)>0$, which is a trivial consequence from \eqref{eq:beta}. There are two different cases:
\par 1. $\alpha\notin(\alpha_1,1-\alpha_0)$:
\begin{align*}
&\frac{\partial \beta_2}{\partial \mu}(n_1,\eta_0,\eta_1,n_2,\alpha,\mu,p)=\frac{\partial \beta}{\partial \mu}(n_1, z_{1-\alpha}\sqrt{\frac{2}{n_1}}, \mu, p)<0,\\
&\frac{\partial \beta_2}{\partial p}(n_1,\eta_0,\eta_1,n_2,\alpha,\mu,p)=\frac{\partial \beta}{\partial p}(n_1, z_{1-\alpha}\sqrt{\frac{2}{n_1}}, \mu, p)<0.
\end{align*}
\par 2. $\alpha\in(\alpha_1,1-\alpha_0)$:
\begin{align*}
& \frac{\partial \beta_2}{\partial \mu}(n_1,\eta_0,\eta_1,n_2,\alpha,\mu,p) = \frac{\partial \beta}{\partial \mu}(n_1,\eta_0,\mu,p)+\int^{\eta_1}_{\eta_0}\frac{\partial \beta}{\partial \eta \partial \mu}(n_1,x_1,\mu,p)\beta\left(n_2,x_2(x_1),\mu,p\right)dx_1 \\
& + \int^{\eta_1}_{\eta_0}\frac{\partial \beta}{\partial \eta}(n_1,x_1,\mu,p)\frac{\partial \beta}{\partial \mu}\left(n_2,x_2(x_1),\mu,p\right)dx_1  = \frac{\partial \beta}{\partial \mu}(n_1,\eta_0,\mu,p)\\
&+\frac{\partial \beta}{\partial \mu}(n_1,x_1,\mu,p)\beta\left(n_2,x_2(x_1),\mu,p\right)\biggr\rvert^{\eta_1}_{\eta_0}+ \frac{n_1}{n_2}\int^{\eta_1}_{\eta_0}\frac{\partial \beta}{\partial \mu}(n_1,x_1,\mu,p)\frac{\partial \beta}{\partial \eta}\left(n_2,x_2(x_1),\mu,p\right)dx_1\\
&+ \int^{\eta_1}_{\eta_0}\frac{\partial \beta}{\partial \eta}(n_1,x_1,\mu,p)\frac{\partial \beta}{\partial \mu}\left(n_2,x_2(x_1),\mu,p\right)dx_1 = \frac{\partial \beta}{\partial \mu}(n_1,\eta_0,\mu,p)\left(1-\beta\left(n_2,x_2(\eta_0),\mu,p\right)\right) \\
&+ \frac{\partial \beta}{\partial \mu}(n_1,\eta_1,\mu,p)\beta\left(n_2,x_2(\eta_1),\mu,p\right) + \frac{n_1}{n_2}\int^{\eta_1}_{\eta_0}\frac{\partial \beta}{\partial \mu}(n_1,x_1,\mu,p)\frac{\partial \beta}{\partial \eta}\left(n_2,x_2(x_1),\mu,p\right)dx_1\\
& + \int^{\eta_1}_{\eta_0}\frac{\partial \beta}{\partial \eta}(n_1,x_1,\mu,p)\frac{\partial \beta}{\partial \mu}\left(n_2,x_2(x_1),\mu,p\right)dx_1 < 0;
\end{align*}
\begin{align*}
& \frac{\partial \beta_2}{\partial p}(n_1,\eta_0,\eta_1,n_2,\alpha,\mu,p) =  \frac{\partial \beta}{\partial p}(n_1,\eta_0,\mu,p)+\int^{\eta_1}_{\eta_0}\frac{\partial \beta}{\partial p\partial \eta}(n_1,x_1,\mu,p)\beta\left(n_2,x_2(x_1),\mu,p\right)dx_1 \\
&+\int^{\eta_1}_{\eta_0}\frac{\partial \beta}{\partial \eta}(n_1,x_1,\mu,p)\frac{\partial \beta}{\partial p}\left(n_2,x_2(x_1),\mu,p\right)dx_1 = \frac{\partial \beta}{\partial p}(n_1,\eta_0,\mu,p)\\
&+\frac{\partial \beta}{\partial p}(n_1,x_1,\mu,p)\beta\left(n_2,x_2(x_1),\mu,p\right)\biggr\rvert^{\eta_1}_{\eta_0} + \frac{n_1}{n_2}\int^{\eta_1}_{\eta_0}\frac{\partial \beta}{\partial p}(n_1,x_1,\mu,p)\frac{\partial \beta}{\partial \eta}\left(n_2,x_2(x_1),\mu,p\right)dx_1\\ 
&+ \int^{\eta_1}_{\eta_0}\frac{\partial \beta}{\partial \eta}(n_1,x_1,\mu,p)\frac{\partial \beta}{\partial p}\left(n_2,x_2(x_1),\mu,p\right)dx_1 =  \frac{\partial \beta}{\partial p}(n_1,\eta_0,\mu,p)\left(1-\beta\left(n_2,x_2(\eta_0),\mu,p\right)\right)\\
&+\frac{\partial \beta}{\partial p}(n_1,\eta_1,\mu,p)\beta\left(n_2,x_2(\eta_1),\mu,p\right) + \frac{n_1}{n_2}\int^{\eta_1}_{\eta_0}\frac{\partial \beta}{\partial p}(n_1,x_1,\mu,p)\frac{\partial \beta}{\partial \eta}\left(n_2,x_2(x_1),\mu,p\right)dx_1 \\
&+ \int^{\eta_1}_{\eta_0}\frac{\partial \beta}{\partial \eta}(n_1,x_1,\mu,p)\frac{\partial \beta}{\partial p}\left(n_2,x_2(x_1),\mu,p\right)dx_1 < 0,
\end{align*}
where $x_2(x_1) = \frac{n_1+n_2}{n_2}\eta_2(n_1, \eta_0, \eta_1, n_2, \alpha)-\frac{n_1}{n_2}x_1$.
\end{proof}

\begin{lemma}The following inequality always holds:
\begin{equation}\label{eq:M1}
1-\beta^{se}_{\text{fw}}(M_1, m) \ge \left(1-\beta_{M_1+1-m}^{se}\right)^{M_1+1-m}.
\end{equation}
If $M_1=M$ and $m=1$, equality is achieved.
\end{lemma}
\begin{proof} Consider an arbitrary set of p-values: $\text{p-value}_1$, \ldots,  $\text{p-value}_M$ of which $M_1$ correspond to the alternatives with the strong effect. Let $\beta_{\text{fw}}(M_1, m)$ be the type II error rate for these p-values. We will show that $$1-\beta_{\text{fw}}(M_1, m) \ge \left(1-\beta_{M_1+1-m}^{se}\right)^{M_1+1-m}.$$ Then, for the set of p-values where the maximum type II error rate is achieved, the statement in the lemma holds.
Let $S_1$ be the set of indexes for the centres with the strong effect. Define $\tilde k$ as the rank of the $m$-th highest p-value among $S_1$,
\begin{equation}\label{eq:def_tildek}
\sum_{c(k)\in S_1}\mathbbm{1}(k\ge \tilde k)=m.
\end{equation}
First, we will prove that
\begin{equation}\label{eq:1}
1-\beta_{\text{fw}}(M_1, m) \ge \mathbb P\left(\text{p-value}(\tilde k)\le\alpha\left(M_1+1-m\right)\right).
\end{equation}
By definition, $\tilde k \ge M_1+1-m$. The condition $\text{p-value}(\tilde k)\le\alpha\left(M_1+1-m\right)$ is sufficient to reject at least $M_1+1-m$ null hypotheses for the centers with strong effect and therefore to avoid type II error. The second inequality we have to prove is
\begin{equation}\label{eq:2}
\mathbb P\left(\text{p-value}(\tilde k)\le\alpha\left(M_1+1-m\right)\right) \ge \left(1-\beta_{M_1+1-m}^{se}\right)^{M_1+1-m}.
\end{equation}
For the proof we will use $S_2$, an arbitrary subset of centers with strong effects such that $|S_2|=M_1+1-m$:
\begin{equation}
\begin{aligned}
& \mathbb P\left(\text{p-value}(\tilde k)\le\alpha\left(M_1+1-m\right)\right) \\
& =  \mathbb P\left(\sum_{c_k \in S_1}\mathbbm 1\left(\text{p-value}(k)\le\alpha\left(M_1+1-m\right)\right)\ge M_1+1-m\right) \\ &\ge \mathbb P\left(\sum_{i \in S_2}\mathbbm 1\left(\text{p-value}_i\le\alpha\left(M_1+1-m\right)\right) \ge M_1+1-m\right) \\& \ge \left(1-\beta_{M_1+1-m}^{se}\right)^{M_1+1-m}.
\end{aligned}
\end{equation}
If $M_1=M$ and $m=1$ then the inequality \eqref{eq:1} becomes an equality, because in this case $\tilde k = M$ and condition $\text{p-value}(M)\le\alpha(M)$ is necessary and sufficient to reject all null hypotheses. Also if $\mathbb P\left(\text{p-value}_i\le\alpha(M)\right)=1-\beta_M^{se}$ for all $i$ then the inequality \eqref{eq:2} becomes an equality:
\begin{equation}
\mathbb P\left(\text{p-value}(M)\le\alpha(M)\right)=\mathbb P\left(\text{p-value}_i\le\alpha(M), \hspace{0.1cm}\forall i=1,\ldots,M\right)=\left(1-\beta_M^{se}\right)^M.
\end{equation} 
\end{proof}
\emph{Comment}: If $M_1<M$ or $m>1$, equality in \eqref{eq:1} cannot be attained, because there is a non-zero probability of having $$\tilde k=M_1+1-m, \hspace{0.2cm} \text{p-value}(\tilde k), \text{p-value}(\tilde k+1) \in\left(\alpha(M_1+1-m),\alpha(M_1+2-m)\right],$$ in this case there is no type II error.

From this lemma, we conclude that $$\beta^{se}_{\text{fw}}(M_1, m)\le 1-\left(1-\beta_{M_1+1-m}^{se}\right)^{M_1+1-m},$$ and $$\beta^{se}_{\text{fw}}(M, 1)=1-\left(1-\beta_{M}^{se}\right)^{M}.$$

\begin{lemma}
If $0<z_{\alpha_0}<z_{1-\alpha_1}$, the maximum probability of conducting the second stage in \eqref{q1max} approximately satisfies
\begin{equation}
\underset{\mu>0, \, p>0}{\max}\left(\beta\left(n_1,\frac{\sqrt 2\, z_{1-\alpha_1}}{\sqrt{n_1}},\mu,p\right)-\beta\left(n_1,\frac{\sqrt 2\,z_{\alpha_0}}{\sqrt{n_1}},\mu,p\right)\right) \approx  2\Phi\left(\frac{z_{1-\alpha_1}-z_{\alpha_0}}{2}\right)-1.
\end{equation}
The maximum is achieved for $\mu = \frac{z_{1-\alpha_1}+z_{\alpha_0}}{\sqrt{2 n_1}},\,p=1$.
\end{lemma}
\begin{proof} To show this, we bound the probability of conduction the second stage as follows:
\begin{align*}
&\underset{\mu>0, \, p>0}{\max}\left(\beta\left(n_1,\frac{\sqrt 2\, z_{1-\alpha_1}}{\sqrt{n_1}},\mu,p\right)-\beta\left(n_1,\frac{\sqrt 2\,z_{\alpha_0}}{\sqrt{n_1}},\mu,p\right)\right) \\
&\approx \Phi\left(\frac{ \sqrt 2 z_{1-\alpha_1}-\sqrt {n_1}\mu p}{\sqrt{2+(1-p)p\mu^2}}\right) - \Phi\left(\frac{ \sqrt 2 z_{\alpha_0}-\sqrt {n_1} \mu p}{\sqrt{2+(1-p)p\mu^2}}\right)\\
& \le \Phi\left(\frac{\sqrt 2 z_{1-\alpha_1} - (\sqrt 2 z_{1-\alpha_1}+\sqrt 2 z_{\alpha_0})/2}{\sqrt{2+(1-p)p\mu^2}}\right) - \Phi\left(\frac{\sqrt 2 z_{\alpha_0} - (\sqrt 2 z_{1-\alpha_1}+\sqrt 2 z_{\alpha_0})/2}{\sqrt{2+(1-p)p\mu^2}}\right) \\&= 2\Phi\left( \frac{\sqrt 2 z_{1-\alpha_1} - \sqrt 2 z_{\alpha_0}}{\sqrt{2+(1-p)p\mu^2}}\right)-1 \le 2\Phi\left(\frac{z_{1-\alpha_1}-z_{\alpha_0}}{2}\right)-1.
\end{align*}
Notice, that an equality is achieved for $\mu = \frac{z_{1-\alpha_1}+z_{\alpha_0}}{\sqrt{2 n_1}}>0$ and $p=1$. This ends the proof.
\end{proof}

%
\section{Type II error rate for the multi-center design computation}\label{Computation of beta fw}
We do not know an exact value of $\beta^{se}_{\text{fw}}(M_1,m)$, but we can compute the type II error rate $\beta_{\text{fw}}(M_1, m)$ for the given responses. Let $\alpha(1)<\alpha(2) < \dots <\alpha(M)$ be the p-value thresholds for the rejection procedure. For convenience put $\alpha(0) = 0,\, \alpha(M+1) = 1$. To compute $\beta_{\text{fw}}(M_1,m)$, one should look at the disposition of the p-values inside the intervals $(\alpha(j-1), \alpha(j)]$, $j=1,2,\ldots,M+1$.

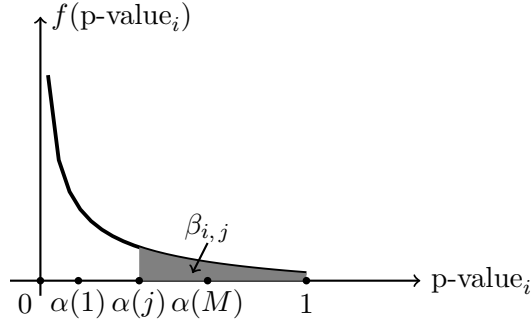
\begin{figure}
\centering
\begin{tikzpicture}
\draw[thick,->] (-0.4,0)--(5,0) node[right]{$\text{p-value}_i$};
\draw[thick,->] (0,-0.2)--(0,3.5) node[right]{$f(\text{p-value}_i)$};
\draw[domain=0.1:3.5,line width=0.5mm] plot (\x,{1/(\x)^(1/2)-0.4345});
\fill [gray, domain=1.3:3.5, variable=\x]
      (1.3, 0)
      -- plot ({\x}, {1/(\x)^(1/2)-0.4345})
      -- (3.5, 0)
      -- cycle;
\draw(-0.2, -0.3) node{$0$};
\draw(0,0) node[inner sep=1pt,fill,circle]{};
\draw(0.5, -0.3) node{$\alpha(1)$};
\draw(0.5,0) node[inner sep=1pt,fill,circle]{};
\draw(1.3, -0.3) node{$\alpha(j)$};
\draw(1.3,0) node[inner sep=1pt,fill,circle]{};
\draw(2.2, -0.3) node{$\alpha(M)$};
\draw(2.2,0) node[inner sep=1pt,fill,circle]{};
\draw(3.5, -0.3) node{$1$};
\draw(3.5,0) node[inner sep=1pt,fill,circle]{};
\draw(2.2, 0.7) node{$\beta_{i,\,j}$};
\draw[thick,->] (2.2, 0.5)--(2,0.1);
\end{tikzpicture}
\caption{$\beta_{i,j}$ for the multiple testing procedure.}\label{Illustration for $i-$th center}
\end{figure}

Define $\beta_{i,j} = \mathbb P(\text{p-value}_i \ge \alpha(j))$ to be the type II error for the $i-$th center evaluated at $\alpha = \alpha(j)$ (see Fig.~\ref{Illustration for $i-$th center}). The probability that the p-value of the $i-$th center is in the interval $(\alpha(j-1), \alpha(j)]$ is
\begin{equation}
\mathbb P(\text{p-value}_i\in (\alpha(j-1), \alpha(j)]) = \beta_{i,j-1}-\beta_{i,j}, \hspace{0.2cm}  i = 1,\ldots, M; \hspace{0.2cm} j = 1,\ldots, M+1.
\end{equation}
Notice that if p-values for some centres are in the same interval $(\alpha(j), \alpha(j-1)]$, the decision about $\mathcal H_0$ there will be the same. 
\par Let the random index $j_i$ define an interval of the $i-$th center p-value,
\begin{equation}
j_i = \left(j \, | \, \text{p-value}_i\in (\alpha(j-1), \alpha(j)] \right),\hspace{0.2cm} j_i = 1, \ldots, M+1.
\end{equation}
The probability of observing $\{j_1, \ldots, j_M\}$ is then expressed as
\begin{equation}\label{eq:block}
 \prod_{i=1}^M\left(\beta_{i, \, j_i-1}-\beta_{i, \,j_i}\right), \hspace{0.2cm} \text{where}\hspace{0.1cm} 1 = \beta_{i,0}>\beta_{i,1}>\ldots>\beta_{i,M}>\beta_{i, M+1}=0.
\end{equation}
Define $e_j$ to be the number of p-values falling in the interval $(0, \alpha(j)]$,
\begin{equation}\label{eq:def_ej}
e_j = e_j(j_1, \ldots, j_M) = \sum_{i=1}^M \mathbbm{1}\left(j_i\le j\right).
\end{equation}
Using $e_j$, an alternative formulation of Hochberg's step-up procedure is as follows:
\begin{itemize}
\item Compute $e_j$ for each $j = 1,\ldots,M$;
\item Find the maximum $J$ such that $e_J \ge J$; 
\item Reject $\mathcal H_0$ in centres $c_k, \, k = 1,\ldots,e_J$.
\end{itemize}
Let $S_1$ be the set of indexes for the centers with the strong effects. Define
\begin{equation}\label{eq:def_tildej}
\tilde j=\min\left(j \, |\sum_{i\in S_1}\mathbbm{1}\left(j_i>j\right)<m\right)
\end{equation}
Now an expression for $\beta_{\text{fw}}(M_1,m)$ can be written as
\begin{equation}\label{eq:def_beta_fw_new}
\beta_{\text{fw}}(M_1,m)=\mathbb P(e_j<j, \hspace{0.1cm}  \text{ for all } j \ge \tilde j).
\end{equation}
Consequently, one can compute the type II error rate as
\begin{equation}\label{eq:beta_fw}
\beta_{\text{fw}}(M_1,m) = \underset{e_j<j \hspace{0.1cm} \text{ for all } j \ge \tilde j}{\left(\sum_{j_1 = 1}^{M+1}\ldots\sum_{j_M = 1}^{M+1}\right)}\prod_{i=1}^M\left(\beta_{i,j_i-1}-\beta_{i,j_i}\right),
\end{equation}
which is the sum of the probabilities of all realisations $\{j_1, \ldots, j_M\}$, such that $e_j <j$  for all $ j \ge \tilde j$.

\bibliography{biblio_paper}%

\begin{thebibliography}{10}

\bibitem{incentive2004}
Evans Barbara~J, Flockhart David~A, Meslin Eric~M. Creating incentives for
  genomic research to improve targeting of therapies.  {\it Nat Med.
  }2004; 10(12):1289--1291.

\bibitem{Wahlgren21}
Wahlgren N.G., Ranasinha K.W., Rosolacci T., et al. Clomethiazole Acute Stroke
  Study (CLASS).  {\it Stroke. }1999; 30(1):21--28.

\bibitem{Rubin:1974aa}
Rubin Donald~B. Estimating causal effects of treatments in randomized and
  nonrandomized studies.{\it Journal of Educational Psychology. }1974; 66(5) p.688.

\bibitem{Frangakis:2002aa}
Frangakis Constantine~E, Rubin Donald~B. Principal stratification in causal
  inference.  {\it Biometrics. }2002; 58(1):21--29.

\bibitem{Muthen:2009aa}
Muth{\'e}n Bengt, Brown Hendricks~C. Estimating drug effects in the presence of
  placebo response: causal inference using growth mixture modeling.  {\it
  Statistics in Medicine. }2009; 28(27):3363--3385.

\bibitem{Leiby:2009aa}
Leiby Benjamin~E, Sammel Mary~D, Ten~Have Thomas~R, Lynch Kevin~G.
  Identification of multivariate responders and non-responders by using
  Bayesian growth curve latent class models.  {\it Journal of the Royal
  Statistical Society: Series C (Applied Statistics). }2009; 58(4):505--524.

\bibitem{He:2014aa}
He~Jiwei, Entsuah Richard. A mixture model using likelihood-based and
  Bayesian approaches for identifying responders and non-responders in
  longitudinal clinical trials.  {\it Pharmaceutical Statistics.
  }2014; 13(5):327--336.
  
  \bibitem{HOCHBERG:1988aa}
Yosef Hochberg.
\newblock A sharper Bonferroni procedure for multiple tests of significance.
\newblock \emph{Biometrika}, 75\penalty0 (4):\penalty0 800--802, 12 1988.

\bibitem{Benjamini:1995aa}
Yoav Benjamini and Yosef Hochberg.
\newblock Controlling the false discovery rate: a practical and powerful
  approach to multiple testing.
\newblock \emph{Journal of the Royal Statistical Society. Series B
  (Methodological)}, pages 289--300, 1995.
  
  \bibitem{essen}
Carl-Gustaf Essen.
\newblock On the liapounoff limit of error in the theory of probability.
\newblock \emph{Arkiv f{\"o}r Matematik, Astronomi Och Fysik},
  1942(A28):\penalty0 1--19, 1942.

\bibitem{Berry:1941aa}
Andrew~C Berry.
\newblock The accuracy of the Gaussian approximation to the sum of independent
  variates.
\newblock \emph{Transactions of the American Mathematical Society}, 49\penalty0
  (1):\penalty0 122--136, 1941.


\bibitem{Agid:2013aa}
Agid Ofer, Siu Cynthia~O, Potkin Steven~G, et al. Meta-regression analysis of
  placebo response in antipsychotic trials, 1970--2010.  {\it American Journal
  of Psychiatry. }2013; 170(11):1335--1344.

\bibitem{Tuttle:2015aa}
Tuttle Alexander~H, Tohyama Sarasa, Ramsay Tim, et al. Increasing placebo
  responses over time in US clinical trials of neuropathic pain.  {\it Pain.
  }2015; 156(12):2616--2626.

\bibitem{Holmes:2016aa}
Holmes RD, Tiwari AK, Kennedy JL. Mechanisms of the placebo effect in pain and
  psychiatric disorders.  {\it The Pharmacogenomics Journal. }2016; 16(6):491.

\bibitem{ARMITAGE:1960aa}
ARMITAGE PMA. Sequential medical trials. {\it Sequential Medical Trials.
  } 1960.

\bibitem{Armitage:1969aa}
Armitage Peter, McPherson CK, Rowe BC. Repeated significance tests on
  accumulating data.  {\it Journal of the Royal Statistical Society. Series A
  (General). }1969; p.235--244.

\bibitem{Siegmund:1978aa}
Siegmund D. Estimation following sequential tests.  {\it Biometrika.
  }1978; 65(2):341--349.

\bibitem{OBrien:1979aa}
O'Brien Peter~C, Fleming Thomas~R. A multiple testing procedure for clinical
  trials.  {\it Biometrics. }1979; p.549--556.

\bibitem{Tsiatis:1984aa}
Tsiatis Anastasios~A., Rosner Gary~L., Mehta Cyrus~R.. Exact Confidence
  Intervals Following a Group Sequential Test. {\it Biometrics.} 1984; 40(3):797--803.

\bibitem{Pavlic:2001aa}
Pavlic Maja, Brand Richard~J, Cummings Steven~R. Estimating probability of
  non-response to treatment using mixture distributions.  {\it Statistics in
  Medicine. }2001; 20(12):1739--1753.

\end{thebibliography}

\Addresses

\end{document}